\documentclass{IEEEtran} 

\usepackage{graphics} 
\usepackage{epsfig} 
\usepackage{amsmath} 
\usepackage{amssymb}  

\usepackage{algpseudocode}
\usepackage{algorithm}
\usepackage{epstopdf}

\usepackage{verbatim}

\usepackage{amsthm}
\usepackage{color}
\usepackage{setspace}
\usepackage{multirow}
\usepackage{multicol}

\newtheorem{assumption}{\bf Assumption}

\newtheorem{theorem}{\bf Theorem}
\newtheorem{proposition}{\bf Proposition}
\newtheorem{lemma}{\bf Lemma}
\newtheorem{corollary}{\bf Corollary}
\newtheorem{remark}{\bf Remark}

\usepackage{enumitem}

\usepackage{hyperref}

\usepackage{tikz}               
\usepgflibrary{arrows}

\begin{document}
\title{Periodic optimal control of nonlinear constrained systems using economic model predictive control}
\author{Johannes K\"ohler$^1$, Matthias A. M\"uller$^2$, Frank Allg\"ower$^1$
\thanks{$^1$Johannes K\"ohler and Frank Allg\"ower are with the
Institute for Systems Theory and Automatic Control, University of Stuttgart,
70550 Stuttgart, Germany. (email:$\{$johannes.koehler, frank.allgower$\}$@ist.uni-stuttgart.de).}
\thanks{$^2$Matthias A. M\"uller is with the Institute of Automatic Control, Leibniz University Hannover, 30167 Hannover, Germany.
(email:mueller@irt.uni-hannover.de).}
\thanks{This work was supported by the German Research Foundation (DFG) under Grants
GRK 2198/1 - 277536708, AL 316/12-2, and MU 3929/1-2 - 279734922.} 
}
\maketitle
\begin{abstract}
In this paper, we consider the problem of periodic optimal control of nonlinear systems subject to online changing and periodically time-varying economic performance measures using model predictive control (MPC). 
The proposed economic MPC scheme uses an online optimized artificial periodic orbit to ensure recursive feasibility and constraint satisfaction despite unpredictable changes in the economic performance index. 
We demonstrate that the direct extension of existing methods to periodic orbits does not necessarily yield the desirable closed-loop economic performance. 
Instead, we carefully revise the constraints on the artificial trajectory, which ensures that the closed-loop average performance is no worse than a locally optimal periodic orbit. 
In the special case that the prediction horizon is set to zero, the proposed scheme is a modified version of recent publications using periodicity constraints, with the important difference that the resulting closed loop has more degrees of freedom which are vital to ensure convergence to an optimal periodic orbit.   
In addition, we detail a tailored offline computation of suitable terminal ingredients, which are both theoretically and practically beneficial for closed-loop performance improvement. 
Finally, we demonstrate the practicality and performance improvements of the proposed approach on benchmark examples.
\end{abstract}
\begin{IEEEkeywords}
Nonlinear model predictive control, economic MPC, dynamic real time optimization,  periodic optimal control, changing economic criteria
\end{IEEEkeywords}
 
\IEEEoverridecommandlockouts
\IEEEpubid{\begin{minipage}{\textwidth}\ \\[18pt] \\ \\
         \copyright 2020 Elsevier Ltd.  All rights reserved.
     \end{minipage}}

\section{Introduction}
Model Predictive Control (MPC) is a well established control method that can cope with nonlinear dynamics, hard state and input constraints, and the inclusion of general performance criteria~\cite{rawlings2017model}. 
Economic MPC~\cite{angeli2012average,ellis2016economic,muller2017economic,faulwasser2018economic} is a variant of MPC that directly aims at improving a user specified economic performance index instead of stabilizing some given setpoint or trajectory. 
In case the system is optimally operated at a given setpoint, there exist well established methods to design economic MPC schemes with closed-loop performance guarantees~\cite{angeli2012average,amrit2011economic,grune2013economic,faulwasser2018economic}.

\subsubsection*{Motivation}
Many practical problems require a paradigm that goes beyond steady-state operation and embraces dynamic operation and online changing conditions. 
Periodic operation is, e.g., economically beneficial (and necessary) in water distribution networks~\cite{limon2014single,wang2018economic}, electrical networks~\cite{pereira2015periodic}, buildings and HVAC~\cite{serale2018model,rawlings2018economic} due to the inherent periodicity (day-night cycle) in changing price signals (supply and demand) or dynamics (outside temperature).   
%
Even in the time-invariant problem of maximizing the production in (nonlinear) continuous stirred-tank reactors (CSTR), periodic operation can be economically beneficial, compare~\cite{bailey1971cyclic,budman2013control}.   
Periodic operations also naturally arise in periodic/cyclic scheduling~\cite{risbeck2019unification} and power generation using kites~\cite{diehl2004efficient}.   
In addition to the challenges related to dynamic/periodic operation, the external operating conditions may change unpredictably and the system is expected to reliably and economically operate despite these changes.  
In this paper, we present an economic MPC framework that provides economic performance guarantees for periodic operation of nonlinear systems. 

\subsubsection*{Related Work}
In~\cite{zanon2017periodic,alessandretti2016convergence,risbeck2019unification} performance guarantees are obtained by using terminal constraints for \textit{a-priori known optimal periodic orbits}.   
Online changes in the optimal system operation cannot be incorporated.

In~\cite{grune2013economic} an economic MPC scheme without terminal ingredients is studied, which has recently been extended to periodic operation~\cite{muller2016economic} and time-varying problems~\cite{grune2018economic}. 
The resulting closed-loop performance guarantees are, however, only valid if a potentially \textit{large prediction horizon} is used. 

In~\cite{houska2017cost,wang2018economic} an economic MPC scheme based on a periodicity constraint is suggested, which minimizes the cost of a periodic orbit starting at the current measured state.
The convergence to the optimal periodic orbit can be theoretically studied based on convergence results for coordinate descent methods~\cite[Lemma~3, Thm.~4]{houska2017cost} or online examined based on dual variables~\cite[Thm.~1]{wang2018economic}. 
Crucially, even in the linear convex case, the closed loop does not necessarily converge to the optimal periodic orbit~\cite[Example~6]{wang2018economic}.

A promising approach to deal with online changing conditions is to simultaneously optimize an artificial reference, which is a well established method in setpoint tracking MPC~\cite{limon2008mpc} and has recently been extended to periodic tracking and nonlinear systems~\cite{limon2016mpc,limon2018nonlinear,JK_periodic_automatica}. 
This idea of using an artificial setpoint in combination with an external update scheme are used in~\cite{muller2013economic,muller2014performance,fagiano2013generalized} to design economic MPC schemes with performance guarantees relative to steady-state operation, compare also~\cite{ferramosca2014economic}.

In~\cite{limon2014single} for linear systems an artificial reference is used to compute the economic optimal periodic orbit and a tracking stage cost is used to ensure stability, compare also~\cite{broomhead2015robust}. 
Recently, in~\cite{gutekunst2020economic} a nonlinear version has been proposed, which also allows to optimize over the period length using a continuous-time formulation. 
Furthermore, instead of a standard tracking stage cost a regularization with respect to the non-periodicity in the input and economic cost is used, which also guarantees convergence to the optimal periodic orbit under appropriate conditions (controllability, no local minima,$\dots$).
However, the usage of a tracking cost/periodic regularization can reduce the transient economic performance, compare, .e.g.~\cite{rawlings2008unreachable}.

 \subsubsection*{Contribution}
We present an economic MPC scheme that ensures recursive feasibility, constraint satisfaction and performance guarantees for nonlinear systems despite unpredictable changes in the economic performance index. 
Recursive feasibility is achieved by including an artificial periodic reference trajectory in the online optimization. 
We use a self-tuning weight for the cost of the artificial reference trajectory in order to obtain suitable bounds on the closed-loop performance, as an extension to~\cite{muller2013economic,muller2014performance}, where this idea was introduced with artificial setpoints. 
We demonstrate by means of a simple motivating example (similar to~\cite{muller2016economic}), that a direct extension of~\cite{muller2013economic,muller2014performance} to the periodic case does not necessarily yield the desired closed-loop performance. 
Instead, we use a novel continuity condition and reformulate the constraints on the optimal periodic orbit  to ensure that:
\begin{enumerate}[label=\alph*)]
\item  the average performance of the artificial reference converges to that of a locally optimal periodic trajectory,
\item  the closed-loop average performance is no worse than that of the (limiting) artificial reference trajectory.
\end{enumerate}
 In the special case of linear systems with convex cost and convex constraints, the closed-loop average performance is no worse than the (globally) optimal periodic orbit. 
Furthermore, if we consider a prediction horizon of $N=0$, we obtain a modified version of recent publications using periodicity constraints~\cite{houska2017cost,wang2018economic}, with the same number of optimization variables and guaranteed convergence to a (local) optimum, which in~\cite{houska2017cost,wang2018economic} can only be ensured under significantly more restrictive conditions. 

Some of the improved performance properties require the usage of suitable terminal ingredients (economic terminal cost and terminal set). 
To this end, we provide a novel design procedure that is applicable to dynamic operation and economic costs, as an extension of the approach in~\cite{amrit2011economic}.

We demonstrate the practicality and performance of the proposed framework using a periodic time-varying  HVAC system~\cite{risbeck2019economic} and a time-invariant nonlinear CSTR~\cite{bailey1971cyclic}. 

\subsubsection*{Outline}
Section~\ref{sec:MPC} introduces the problem setup and demonstrates that the periodic case requires additional care with a simple system. 
The proposed economic MPC framework with theoretical analysis is presented in Section~\ref{sec:main}. 
Additional details and variations are discussed in Section~\ref{sec:ext}. 
Section~\ref{sec:num} illustrates the results with numerical examples.
Section~\ref{sec:sum} concludes the paper.  
Appendix~\ref{app:CSTR} contains additional details regarding the CSTR example.

\subsubsection*{Notation}
The identity matrix is $I\in\mathbb{R}^{n\times n}$.  
The interior of a set $\mathcal{X}\subset\mathbb{R}^n$ is $\text{int}(\mathcal{X})$.
A ball with radius $\epsilon$ around a point $x\in\mathbb{R}^n$ is $\mathbb{B}_{\epsilon}(x)=\{y\in\mathbb{R}^n|~\|x-y\|\leq \epsilon\}$. 
By $\mathcal{K}_{\infty}$ we denote the class of functions $\alpha:\mathbb{R}_{\geq 0}\rightarrow\mathbb{R}_{\geq 0}$, which are continuous, strictly increasing, unbounded and satisfy $\alpha(0)=0$. 
\section{Periodic economic model predictive control}
\label{sec:MPC}
In Section~\ref{sec:setup} we introduce the problem setup. 
Some existing economic MPC methods for this problem setup are briefly discussed in Section~\ref{sec:methods}. 
We demonstrate the potential difficulties in periodic problems (compared to steady-state) with a simple system in Section~\ref{sec:pitfal}.
%
\subsection{Problem setup}
\label{sec:setup}
We consider periodic problems with a fixed known period length $T\in\mathbb{N}$. 
For many systems (HVAC, water distribution networks) this periodicity is inherent to the problem setup (in the dynamics and/or cost function). 
In time-invariant problems (chemical reactor) this period length $T$ is a user specified decision variable that influences the possible performance improvement (compared to the steady-state operation).  
Both cases are illustrated in the numerical examples in Section~\ref{sec:num}. 

We consider nonlinear periodically time-varying discrete-time systems
\begin{align}
\label{eq:sys}
x(t+1)&=f(x(t),u(t),t)
\end{align}
 with the state $x\in\mathbb{R}^n$, control input $u\in\mathbb{R}^m$, and time step $t\in\mathbb{N}$. 
We assume that the dynamics are periodic with the (known) period length $T$, i.e.,  $f(x,u,t)=f(x,u,t+T)$. 
We impose point-wise in time constraints on the state and input
$(x(t),u(t))\in \mathcal{Z}(t)$, 
with compact periodically time-varying sets $\mathcal{Z}(t)\subset\mathbb{R}^{n+m}$, i.e., $\mathcal{Z}(t)=\mathcal{Z}(t+T)$. 
We consider reference constraint sets $r(t)=(x_r(t),u_r(t))\in\mathcal{Z}_r(t)$, that satisfy $\mathcal{Z}_r(t)=\mathcal{Z}_r(t+T)$ and $\mathcal{Z}_r(t)\subseteq\text{int}(\mathcal{Z}(t))$ for all $t\geq 0$. 
We define the set of feasible $T$-periodic reference trajectories $r_T=(x_{r_T},u_{r_T})\in\mathbb{R}^{(n+m) T}$ as 
\begin{align*}
\mathcal{Z}_T(t)=&\{r_T=(x_{r_T},u_{r_T})|\\
& x_{r_T}(0)=f(x_{r_T}(T-1),u_{r_T}(T-1),t+T-1),\\
&x_{r_T}(k+1)=f(x_{r_T}(k),u_{r_T}(k),t+k),\quad k=0,\dots, T-2,\\
&r_T(k)\in\mathcal{Z}_r(t+k),\quad k=0,\dots, T-1\}.
\end{align*} 
We define a periodic shift matrix $\mathcal{R}_T\in\mathbb{R}^{(n+m)T\times (n+m)T}$, which satisfies $r_T\in\mathcal{Z}_T(t)~\Rightarrow~\mathcal{R}_T r_T\in\mathcal{Z}_T(t+1)$ and $\Pi_{j=0}^{T-1}\mathcal{R}_T=\mathcal{R}_T^T=I$. 
Whenever clear from the context, we denote the first element of the periodic reference $r_T$ by $(x_r,u_r)=r=r_T(0)$.

The economic performance measure is given by a general (non-convex) periodically time-varying function 
\begin{align*}
\ell(x,u,t,y)=\ell(x,u,t+T,y),
\end{align*}
which can depend on external parameters $y\in\mathbb{Y}$, with $\mathbb{Y}$ compact. 
At each time step $t$, the parameters $y(t)\in\mathbb{Y}$ are assumed to be available as an external (user defined) input. 
These parameters might incorporate online changing prices and/or general changes in the desired production/operation. 
For simplicity, we assume that $\ell$ and $f$ are continuous, which (in combination with compact constraints) implies that $\ell$ is bounded. 
\begin{remark}
\label{rk:price}
We consider the setting with a constant parameter $y(t)$ in the predictions to simplify the notation. 
However, the presented guarantees hold equally if we consider a predicted (periodic) sequence of parameters $y(\cdot|t)$. 
In case some of the constraints in $\mathcal{Z}$, $\mathcal{Z}_r$ are relaxed to soft constraints using penalty terms in the stage cost $\ell$, the external parameters $y$ can also model online unpredictably changing constraints sets.
\end{remark}

An optimal $T$-periodic orbit at time $t$ is the solution to the following optimization problem
\begin{align}
\label{eq:opt_periodic}
\min_{r_T(\cdot|t)\in\mathcal{Z}_T(t)}&J_T(r_T(\cdot|t),t,y(t))
:=\sum_{j=0}^{T-1} \ell (r_T(j|t),t+j,y(t))
\end{align}
and is denoted by $\overline{r}_T(\cdot|t)$. 
If the external parameters remain constant, i.e. $y(t+1)=y(t)$, then an optimal periodic trajectory at the next time step is given by $\overline{r}_T(\cdot|t+1)=\mathcal{R}_T\overline{r}_T(\cdot|t)$, i.e., $\overline{r}_T(1|t)=\overline{r}_T(0|t+1)$. 

Given some initial state $x(0)$, the closed-loop average economic cost is defined as
\begin{align}
\label{eq:av_cost}
\overline{J}_{cl}(x(0))=\limsup_{K\rightarrow\infty} \dfrac{1}{K}\sum_{k=0}^{K-1} \ell(x(k),u(k),k,y(k)).
\end{align}
The control goal is to minimize this closed-loop average cost~\eqref{eq:av_cost} and achieve constraint satisfaction  $(x(t),u(t))\in\mathcal{Z}(t)$, for all $t\geq 0$.

%
\subsection{Existing methods}
\label{sec:methods} 
One way to approach this problem, would be to solve~\eqref{eq:opt_periodic} offline to obtain $\overline{r}_T$ and then use terminal constraints, similar to~\cite{zanon2017periodic,alessandretti2016convergence,risbeck2019unification}. 
If the parameters $y$ stay constant, this strategy ensures the following performance bound 
\begin{align}
\label{eq:performance_standard}
\overline{J}_{cl}(x(0))\leq J_T(\overline{r}_T)/T. 
\end{align}
However, if the parameters $y$ change online the performance deteriorates. 
Directly recomputing the optimal periodic orbit $r_T$ and adjusting the terminal constraints based on online changes in the parameters $y$ can cause feasibility issues. 
If there exists only a finite set of possible parameter values $y\in\mathbb{Y}$, feasible transition trajectories can be computed offline to avoid such issues, compare~\cite{angeli2016theoretical}. 

The issue of recomputing the optimal orbit $\overline{r}_T$ under changing parameters $y$ can be avoided by using economic MPC schemes without any terminal constraints~\cite{grune2013economic,muller2016economic,grune2018economic} which implicitly find the optimal mode of operation (using turnpike arguments). 
The corresponding theoretical properties may require a-priori assumptions on the optimal mode of operation  (dissipativity, turnpike, \textit{optimal}\footnote{In some cases, a multi-step MPC scheme based on the optimal period length $T$ needs to be implemented, compare~\cite{muller2016economic}.} period length $T$), which can be difficult to verify for practical systems.  
In addition, a potentially very large prediction horizon $N$ may be required to ensure these properties.  

A reliable method to deal with these issues is to simultaneously optimize an artificial periodic reference $r_T\in\mathcal{Z}_T$ in the MPC problem and use a tracking formulation with terminal constraints to stabilize this artificial reference trajectory, compare~\cite{limon2014single} and \cite{limon2008mpc,limon2016mpc,limon2018nonlinear,JK_periodic_automatica}.  
This direct stabilization, however, does not take into account the potential for economic performance improvement and may hence result in severe suboptimality, compare~\cite{rawlings2008unreachable} and~\cite[Sec.~3.4]{faulwasser2018economic}.  

Thus, to ensure theoretical properties (recursive feasibility, performance bounds) and allow for online performance improvement, we will use an artificial periodic reference trajectory $r_T\in\mathcal{Z}_T$ and combine it with a purely economic formulation (without using any tracking costs), as an extension to~\cite{muller2013economic,muller2014performance,fagiano2013generalized}. 
In the numerical example in Section~\ref{sec:CSTR}, we demonstrate the advantages of such a formulation compared to some of the existing methods. 

%
\subsection{Pitfalls - Generalized periodic constraints}
\label{sec:pitfal}
In the following, we show that a direct (naive) extension of existing generalized terminal setpoint constraints in~\cite{muller2013economic,muller2014performance,fagiano2013generalized} to periodic reference trajectories $r_T$, does not necessarily imply the desirable economic performance guarantees and thus requires further modifications (which will be introduced in Section~\ref{sec:scheme}).

Consider a scalar time-invariant system with $f(x,u)=u$, $x\in\{0,1,2\}$ which is depicted as a graph in Figure~\ref{fig:graph} with some arbitrary small positive constant $\epsilon$, similar to~\cite[Example~4]{muller2016economic}. 
The optimal periodic orbit is $x_{\overline{r}_T}=(1,2)$  with cost $J_T(\overline{r}_T)=\epsilon$ and $T=2$. 
\begin{figure}[hbtp]
\begin{center}
\begin{tikzpicture}[->,>=stealth',shorten >=1pt,auto,node distance=2.3cm,
  thick,main node/.style={circle,draw,fill=blue!15,font=\sffamily\large}]
  \node[main node] (0) {$0$};
  \node[main node] (1) [right of=0] {$1$};
  \node[main node] (2) [right of=1] {$2$};
  \path[every node/.style={font=\sffamily\small}]
    (0) edge [loop left] node[left] {$u=0 \atop \ell(x,u)=1$} (0)
    (0)  edge [bend right] node[above] {${u=1 \atop \ell(x,u)=0}$} (1)
    (1)  edge [bend right] node[below] {${u=2 \atop \ell(x,u)=1+\epsilon}$} (2)
    (2)  edge [bend right] node[above] {${u=1 \atop \ell(x,u)=-1}$} (1);
\end{tikzpicture}
\end{center}
\caption{Academic counter example - Illustration of feasible transitions. } 
\label{fig:graph}
\end{figure}
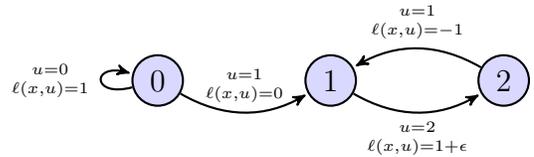

The following economic MPC scheme with an artificial periodic trajectory $r_T$ can be viewed as a generalization of the steady-state methods in~\cite{muller2013economic,muller2014performance,fagiano2013generalized}:
\begin{subequations}
\label{eq:limon_eco_fail}
\begin{align}
&\min_{u(\cdot|t),r_T(\cdot|t)}\sum_{k=0}^{N-1}\ell(x(k|t),u(k|t))+ J_{T}(r_T(\cdot|t))\\
\text{s.t. }&x(k+1|t)=f(x(k|t),u(k|t)),\quad x(0|t)=x(t),\\
&(x(k|t),u(k|t))\in\mathcal{Z},\quad k=0,\dots,N-1,\\  
&r_T(\cdot|t)\in\mathcal{Z}_T,\quad x(N|t)=x_{r_T}(0|t).  
\end{align}
\end{subequations}
This optimization problem computes an open-loop trajectory $x^*(\cdot|t)$ starting at $x(t)$ that ends on some periodic trajectory $r_T\in\mathcal{Z}_T$. 
In closed-loop operation, the optimization problem~\eqref{eq:limon_eco_fail} would be solved in each time step $t$ and the first part of the optimized input trajectory applied to the system, i.e., $u(t)=u^*(0|t)$, $x(t+1)=x^*(1|t)$. 
Although the steady-state schemes~\cite{muller2013economic,muller2014performance,fagiano2013generalized} often have additional modifications (terminal cost, self-tuning weights, additional constraints on $r_T$), the problem we discuss in the following remains the same. 

Consider the initial condition $x(0)=0$ and a prediction horizon of $N=2$. 
The artificial reference is the optimal periodic orbit $r_T(\cdot|t)\in\{1,2\}$. 
The only feasible trajectories that satisfy $x(N|t)=x(2|t)\in\{1,2\}$ are $u(\cdot|t)=(0,1)$ and $u(\cdot|t)=(1,2)$, and the corresponding open-loop cost is $1+0$ and $0+1+\epsilon$, respectively. 
Thus, the optimal solution to~\eqref{eq:limon_eco_fail} satisfies $x^*(1|t)=0=x(t)$. 
Correspondingly, the closed-loop system based on~\eqref{eq:limon_eco_fail} stays at $x(t)=0$ and encounters the economic cost $\ell(x(t),u(t))=1$, $\forall t\geq 0$ and does not achieve the same performance as the artificial periodic reference $r_T$. 
This issue can persist, even if we choose an arbitrarily large (even) prediction horizon $N$. 
In particular, with~\eqref{eq:limon_eco_fail}, we can only ensure 
\begin{align*}
\overline{J}_{cl}(x(0))\leq \max_{k}\ell(r_T(k))=1+\epsilon. 
\end{align*}
This is in contrast to existing results for the steady-state case ($T=1$)~\cite{muller2013economic,muller2014performance,fagiano2013generalized}, which can ensure the superior bound~\eqref{eq:performance_standard}. 
The same problem appears in economic MPC schemes without terminal constraints, compare~\cite[Examples 4 and 18]{muller2016economic}. 
One way to alleviate this problem is to apply the first $T$ components of the open-loop input sequence $u^*(\cdot|t)$ (multi-step MPC)~\cite{muller2016economic,wabersich2018economic}, which transforms the problem to a higher dimensional steady-state problem ($T$-step system~\cite{muller2015role,muller2016economic,wabersich2018economic}). 
Since we wish to consider problems with possibly large period lengths $T$, this solution seems, however, inadequate. 
If we would use an economic MPC scheme based on periodicity constraints~\cite{houska2017cost,wang2018economic}, the closed-loop system would also stay at $x(t)=0$ for all $ t\geq 0$, since there exists only one feasible periodic orbit starting at $x(0)=0$. 
The theoretical results in~\cite{houska2017cost} do not apply, since the one-step controllability condition~\cite[Lemma~4]{houska2017cost} is not satisfied.

To summarize, as also discussed in \cite{risbeck2019unification}, the existing approaches with online optimized periodic reference trajectory $r_T$ \textit{do not come with any closed-loop performance guarantees similar to~\eqref{eq:performance_standard}}.

\section{Proposed periodic economic MPC framework}
\label{sec:main}
This section contains the main result of the paper. 
The proposed scheme is detailed in Section~\ref{sec:scheme}. 
Performance guarantees relative to the limiting artificial reference trajectory are derived in Section~\ref{sec:theory_1}. 
In Section~\ref{sec:theory_2}, improved a priori performance bounds are derived based on the self-tuning weight. 
The theoretical properties are summarized in Theorem~\ref{thm:main}. 
\subsection{Proposed scheme}
\label{sec:scheme} 
In the following, we detail the proposed scheme and discuss the relation to other existing methods. 
The main idea is to directly minimize the predicted economic stage cost $\ell$ with some terminal cost $V_f$ and terminal set $\mathcal{X}_f$ around the artificial reference trajectory $r_T$, and use an updating scheme to ensure that the optimized artificial reference trajectory $r_T$ converges to the best possible periodic orbit $\overline{r}_T$. 
\subsubsection*{Optimization problem}
The corresponding optimization problem at each time step $t$ is given by 
\begin{subequations}
\label{eq:MPC}
\begin{align}
\label{eq:MPC_pred_eco_cost}
&\min_{u(\cdot|t),r_T(\cdot|t)}\sum_{k=0}^{N-1}\ell(x(k|t),u(k|t),t+k,y(t))\\
\label{eq:MPC_term_cost}
&+V_{f}(x(N|t),r_T(\cdot|t),t+N,y(t))\\
\label{eq:MPC_selftune_cost}
&+\beta(t)\cdot J_{T}(r_T(\cdot|t),t+N,y(t))\\
\label{eq:MPC_dyncon}
\text{s.t. }&x(k+1|t)=f(x(k|t),u(k|t),t+k),~ x(0|t)=x(t),\\
\label{eq:MPC_cons}
&(x(k|t),u(k|t))\in\mathcal{Z}(t+k),\quad k=0,\dots,N-1,\\ 
\label{eq:MPC_term}
& x(N|t)\in\mathcal{X}_f(r_T(\cdot|t),t+N),\\
\label{eq:MPC_periodic}
&r_T(\cdot|t)\in\mathcal{Z}_T(t+N),\\
\label{eq:kappa_con_1}
&\Delta \kappa(t)=\sum_{j=0}^{T-1}\left[ \ell(r_T(j|t),t+N+j,y(t))-\kappa_j(t)\right],\\
\label{eq:kappa_con_2}
&\ell(r_T(j|t),t+N+j,y(t))\leq \kappa_j(t)-c_{\kappa}\Delta \kappa(t) ,\\
& j=0,\dots,T-1, \nonumber
\end{align}
\end{subequations}
with some positive constant $c_{\kappa}$. 
The solution to~\eqref{eq:MPC} are optimal state and input trajectories $x^*(\cdot|t)$, $u^*(\cdot|t)$ and an artificial periodic reference trajectory $r_T^*(\cdot|t)$. 
The input trajectory minimizes the predicted economic cost~\eqref{eq:MPC_pred_eco_cost} with a terminal cost~\eqref{eq:MPC_term_cost}, to be specified later (Ass.~\ref{ass:term_simple}).  
The economic cost of the artificial periodic reference trajectory $r_T$ is weighted with a self-tuning (time-varying) weight $\beta(t)$~\eqref{eq:MPC_selftune_cost}, similar to~\cite{muller2013economic,muller2014performance}. 
The resulting state and input trajectory satisfy the dynamics~\eqref{eq:MPC_dyncon} and the posed state and input constraints~\eqref{eq:MPC_cons}. 
In addition, the terminal state of the predicted state sequence is in a terminal set (see Ass.~\ref{ass:term_simple} below) around the artificial reference trajectory~\eqref{eq:MPC_term}. 
The artificial reference is a feasible periodic orbit~\eqref{eq:MPC_periodic}. 
Conditions~\eqref{eq:kappa_con_1}--\eqref{eq:kappa_con_2} pose additional constraints on the improvement of the economic cost of the artificial reference $r_T$ compared to $\kappa_j$, similar to~\cite{muller2013economic,muller2014performance,fagiano2013generalized}. 
In particular, if $\Delta \kappa$ is negative (the cost $J_T$ of the reference improves), then $\ell(r_T(j|t),t+N+j,y(t))$ can be larger than $\kappa_j(t)$. 
Hence, the constraint~\eqref{eq:kappa_con_2} is less restrictive than $\ell(r_T(j|t),t+N+j,y(t))\leq \kappa_j(t)$. 
The memory states $\kappa_j$ in combination with the self-tuning weight $\beta$ and the constant $c_{\kappa}$ (Ass.~\ref{ass:cont_orbit}) are crucial to establish the desired performance guarantees and are discussed in more detail in the following theoretical analysis. 
For notational simplicity, we define $r_T(T|t):=r_T(0|t)$.  
\subsubsection*{Closed-loop operation}
The closed-loop system with~\eqref{eq:MPC} is given by
\begin{subequations}
\label{eq:close}
\begin{align}
\label{eq:close_1}
x(t+1)=&f(x(t),u^*(0|t),t),\\
\label{eq:close_2}
\beta(t+1)=&\mathcal{B}(\beta(\cdot),\kappa(\cdot),x(\cdot)),\\
\label{eq:close_3}
\kappa_j(t+1)=& \ell(r^*_T(j+1|t),t+N+1+j,y(t+1)),\\
& j=0,\dots,T-1.\nonumber
\end{align}
\end{subequations}
At each time $t$, we measure a state $x(t)$ and an external parameter $y(t)$. 
As in a standard MPC framework, we apply the first part of the optimized input sequence~\eqref{eq:close_1}.  
The cost of the last optimal artificial reference trajectory $r_T^*$ is saved in the memory states $\kappa_j$~\eqref{eq:close_3}.
The tuning weight $\beta(t)$ can be determined by some general (causal) update rule $\mathcal{B}$~\cite{muller2013economic,muller2014performance}, or simply chosen by a user as a time-varying or constant signal (c.f.~\cite[Update rule~1]{muller2013economic} and Sec.~\ref{sec:const_beta}).
The constant $c_\kappa>0$  and the terminal ingredients $V_f,\mathcal{X}_f$ need to be designed offline, which will be detailed later. 
Hence, compared to a standard MPC closed loop, we have $T$ additional scalar memory states $\kappa_j$~\eqref{eq:close_3} and one scalar tuning-variable $\beta$~\eqref{eq:close_2}.
\subsubsection*{Initialization}
We assume that the initial state $x(0)$ is such that there exists a feasible trajectory to some feasible periodic orbit $r_T$. 
The memory states $\kappa_j(0)$ can be initialized arbitrarily large, such that the constraint~\eqref{eq:kappa_con_2} is inactive at $t=0$. 
Correspondingly, the optimization problem~\eqref{eq:MPC} is feasible at $t=0$. 
The tuning variable $\beta$ can be initialized with any positive scalar, most naturally $\beta(0)=1$.

\subsubsection*{Existing schemes = special cases}
In the following, we discuss in detail how various existing methods for economic MPC are contained in this formulation as special cases. 

The proposed formulation can best be viewed as an extension to~\cite{muller2013economic,muller2014performance,fagiano2013generalized}, which considers an artificial reference setpoint ($T=1$). 
In particular, if we assume a time-invariant problem setup and choose $T=1$, we get the optimization problem and closed-loop operation in~\cite{muller2013economic,muller2014performance,fagiano2013generalized}. 
For $c_{\kappa}\geq 0$, the constraints~\eqref{eq:kappa_con_1}--\eqref{eq:kappa_con_2} are equivalent to $\ell(r_T(t))\leq\kappa(t)=\ell(r_T^*(t-1))$ which is used in~\cite{muller2013economic,muller2014performance,fagiano2013generalized} to ensure that the cost of the artificial reference $r$ is improving. 
Although one can directly see that~\cite{muller2013economic,muller2014performance,fagiano2013generalized} is a special case of the posed formulation, it is not obvious from the onset that the extension of~\cite{muller2013economic,muller2014performance,fagiano2013generalized} to periodic problems should be given by the optimization problem~\eqref{eq:MPC}. 
A more intuitive extension might be to use the constraint $J_T(r_T(\cdot|t))\leq \kappa(t)=J_T(r_T^*(\cdot|t-1))$ (as an alternative to~\eqref{eq:kappa_con_1}--\eqref{eq:kappa_con_2}).  
The possibly suboptimal performance of such an approach has, however, been illustrated in Section~\ref{sec:pitfal}.  
In Section~\ref{sec:term_cost_mod} we show that we can guarantee the same propertiesIn Section~\ref{sec:term_cost_mod} we show that we can guarantee the same properties with this more intuitive extension, if we instead suitably reformulate the cost function. 
Another possible formulation for periodic orbits would be the constraint $\ell(r_T(j|t))\leq \kappa_j(t)$ without the additional term $c_{\kappa}$ in~\eqref{eq:kappa_con_2}. 
This modification is sufficient to avoid the pitfall in Section~\ref{sec:pitfal}, if the artificial reference is initialized as an optimal periodic orbit $\overline{r}_T$. 
 However, this more restrictive constraint can potentially prevent the artificial reference trajectory $r_T$ to converge to the optimal periodic orbit $\overline{r}_T$, compare also Ass.~\ref{ass:cont_orbit} and the numerical example in Section~\ref{sec:num}.  

If we consider a prediction horizon of $N=0$ and a terminal equality constraint $\mathcal{X}_f(r_T,t)=\{x_{r}\}$, then the proposed formula yields a modified version of the MPC scheme using periodicity constraints~\cite{houska2017cost,wang2018economic}.  
The only difference would be the additional performance constraints on the periodic orbit~\eqref{eq:kappa_con_1}--\eqref{eq:kappa_con_2}, which may not be necessary in many cases, compare Section~\ref{sec:term_cost_mod}. 
Crucially, if we choose a suitable terminal cost with a non-empty terminal set (Ass.~\ref{ass:term}), then we can establish closed-loop performance guarantees (Thm.~\ref{thm:main}), which are in general not valid for MPC schemes using periodicity constraints~\cite{houska2017cost,wang2018economic}. 
In particular, the terminal ingredients (Ass.~\ref{ass:term}) relax the one-step controllability condition~\cite[Lemma~4]{houska2017cost} to an incremental stabilizability condition (Ass.~\ref{ass:term}). 
In Lemma~\ref{lemma:change_r_TEC}, we discuss how to retain these performance guarantees with a terminal equality constraint.

The tracking MPC scheme~\cite{limon2014single} can be viewed as a modified version, which uses a tracking stage cost $\ell$ in the optimization problem~\eqref{eq:MPC} and hence does not require the additional tuning variable~$\beta$ or memory state $\kappa_j$.

The standard economic MPC formulations for periodic orbits~\cite{zanon2017periodic,alessandretti2016convergence} and steady-states~\cite{angeli2012average,amrit2011economic} are contained as a special case, if we fix the artificial reference trajectory $r_T=\overline{r}_T$.

\subsection{Relative performance guarantees}
\label{sec:theory_1}
In Proposition~\ref{prop:feas}, we show that the proposed formulation is recursively feasible. 
For constant parameters $y$, Proposition~\ref{prop:av_perf} shows that the average closed-loop performance is no worse than the performance of the limiting artificial references.

\subsubsection*{Terminal ingredients}
The following assumption captures the (standard) sufficient conditions for the terminal ingredients. 
\begin{assumption}
\label{ass:term_simple}
There exists a terminal set $\mathcal{X}_f(r_T,t)$, a (bounded) terminal cost $V_f(x,r_T,t,y)$ and terminal controller $k_f(x,r_T,t)$, such that at any time $t\in\mathbb{N}$, for any parameters $y\in\mathbb{Y}$, reference $r_T\in\mathcal{Z}_T(t)$ and any $x\in\mathcal{X}_f(r_T,t)$, the following conditions hold
\begin{subequations}
\label{eq:term_simple}
\begin{align}
\label{eq:term_PI}
x^+\in&\mathcal{X}_f(r_T^+,t+1),\\
\label{eq:term_con}
(x,u)\in&\mathcal{Z}(t),\\
\label{eq:term_dec}
&V_f(x^+,r_T^+,t+1,y)-V_f(x,r_T,t,y)\\
\leq &-\ell(x,u,t,y)+\ell(r_T(0),t,y),\nonumber
\end{align}
\end{subequations}
with $x^+=f(x,u,t)$, $r_T^+=\mathcal{R}_Tr_T\in\mathcal{Z}_T(t+1)$, $u=k_f(x,r_T,t)$.  
\end{assumption}
This assumption can always be satisfied by using a terminal equality constraint $\mathcal{X}_f(r_T,t)=\{x_r\}$, $k_f=u_r$, with $(x_r,u_r)=r_T(0)$.
However, for the improved performance guarantees discussed in Section~\ref{sec:theory_2} we will require stronger conditions for the terminal ingredients (Ass.~\ref{ass:term}, Sec.~\ref{sec:term}), compare~\cite{muller2014performance}.

\subsubsection*{Recursive feasibility}
The following proposition shows that feasibility of the proposed scheme is independent of the exogenous parameters~$y$.  
\begin{proposition}
\label{prop:feas}
Let Assumption~\ref{ass:term_simple} hold and assume that~\eqref{eq:MPC} is feasible at $t=0$. 
Then the optimization problem~\eqref{eq:MPC} is recursively feasible for the resulting closed-loop system~\eqref{eq:close}. 
\end{proposition}
\begin{proof}
This result is a straightforward extension of existing results for MPC with artificial reference trajectories~\cite{limon2008mpc,limon2016mpc,limon2018nonlinear,JK_periodic_automatica,muller2013economic,muller2014performance,fagiano2013generalized,ferramosca2014economic,houska2017cost,wang2018economic,limon2014single}. 
Given the feasible reference $r_T^*(\cdot|t)\in\mathcal{Z}_T(t+N)$ at time $t$, the shifted reference $r_T(\cdot|t+1)=\mathcal{R}_Tr_T^*(\cdot|t)$ satisfies~\eqref{eq:MPC_periodic}. 
This reference satisfies the constraints~\eqref{eq:kappa_con_2} with equality, since $\Delta \kappa(t+1)=0$. 
A corresponding candidate input sequence is given by
\begin{align*}
u(k|t+1)=\begin{cases}
u^*(k+1|t)&k\leq N-2\\
k_f(x^*(N|t),r_T^*(\cdot|t),t+N)&k=N-1
\end{cases}. 
\end{align*}
The resulting state and input sequences satisfy the constraints~\eqref{eq:MPC_cons} and the terminal constraint~\eqref{eq:MPC_term} due to Ass.~\ref{ass:term_simple}.  
\end{proof}

\subsubsection*{Self-tuning weight}
Define the change in the weight $\beta$ as $\gamma(t)=\beta(t+1)-\beta(t)$. 
The following assumption characterizes some of the properties the update scheme $\mathcal{B}$~\eqref{eq:close_2} should have, such that performance guarantees hold despite online changing values of $\beta$, compare~\cite{muller2013economic} for a more nuanced discussion and an alternative condition on $\mathcal{B}$ resulting in slightly weaker performance guarantees. 
\begin{assumption}
\label{ass:B1}\cite[Ass.~1]{muller2013economic} 
The sequence $\beta(\cdot)$ satisfies  $\limsup_{t\rightarrow \infty}\gamma(t)\leq 0$ and $\gamma(t)\leq c_{\gamma}$, $\beta(t)\geq 0$ for all $t\geq 0$ with a constant $c_{\gamma}$.
\end{assumption} 
Define the cost of the artificial reference as  
\begin{align}
\label{eq:kappa_def}
\kappa(t)=\sum_{j=0}^{T-1}\kappa_j(t)=J_T(r_T^*(\cdot|t-1),t+N-1,y(t)).
\end{align}
Suppose that the parameters $y(t)$ remain constant, 
then the conditions~\eqref{eq:kappa_con_2},\eqref{eq:close_3} with $c_{\kappa}\geq 0$ ensure that $\Delta \kappa(t)=\kappa(t+1)-\kappa(t)\leq 0$ and thus $\kappa$ is non-increasing. 
Boundedness of $\ell$ implies boundedness of $\kappa(t)$. 
Thus $\kappa(t)$ converges to some limit $\kappa_{\infty}$.

\subsubsection*{Average performance}
The following proposition establishes that the closed-loop performance is no worse than $\kappa_{\infty}$ (the performance of the limiting artificial trajectories $r_T$), as an extension to~\cite[Thm.~1]{muller2013economic}. 
\begin{proposition}
\label{prop:av_perf}
Let Assumptions~\ref{ass:term_simple} and \ref{ass:B1} hold and assume that ~\eqref{eq:MPC} is feasible at $t=0$ and  $y(t)$ is constant, 
then the resulting closed-loop system~\eqref{eq:close} satisfies the following performance bound
\begin{align}
\label{eq:performance_bound1}
\limsup_{K\rightarrow\infty}\dfrac{\sum_{t=0}^{TK-1}\ell(x(t),u(t),t,y(t))}{TK}\leq \kappa_{\infty}/T.
\end{align}
\end{proposition}
\begin{proof}
Define the value function
\begin{align}
\label{eq:W}
W(t)=&W(x(t),y(t),\beta(t),\kappa_j(t))\\
=&\sum_{k=0}^{N-1} \ell(x^*(k|t),u^*(k|t),t+k,y(t))\nonumber\\
&+V_f(x^*(N|t),r_T^*(\cdot|t),t+N,y(t))\nonumber\\
&+\beta(t)J_T(r_T^*(\cdot|t),t+N,y(t)).\nonumber
\end{align}	
Proposition~\ref{prop:feas} provides a feasible candidate solution $u(\cdot|t+1)$, $r_T(\cdot|t+1)$ to the optimization problem~\eqref{eq:MPC} at time $t+1$. 
Hence, we can use the cost of the candidate solution to upper bound the value function $W(t+1)$, which in combination with $y(t)$ constant and the terminal cost (Ass.~\ref{ass:term_simple}) yields
\begin{align}
\label{eq:W_1}
&W(t+1)-W(t)+\ell(x(t),u(t),t,y(t))\\
\leq& \ell(r_T^*(0|t),t+N,y(t))+\gamma(t)\kappa(t+1)\nonumber\\
\stackrel{\eqref{eq:close_3}}{=} &\kappa_{T-1}(t+1)+\gamma(t)\kappa(t+1),\nonumber
\end{align}
compare~\cite[Thm.~1]{muller2013economic}\cite[Thm.~1]{muller2014performance} for details. 
The definition of $\kappa_j$ in~\eqref{eq:close_3}, constant parameters $y(t)$ and the constraints~\eqref{eq:kappa_con_1},\eqref{eq:kappa_con_2} ensure 
\begin{align}
\label{eq:kappa_bound_0}
\kappa_j(t+1)\stackrel{\eqref{eq:close_3}}=&\ell(r_T^*(j+1|t),t+N+j+1,y(t))\nonumber\\
\stackrel{\eqref{eq:kappa_con_2}}{\leq} &\kappa_{j+1}(t)-c_{\kappa}\Delta \kappa(t),\quad j=0,\dots,T-1,
\end{align}
 with $\kappa_T:=\kappa_0$. Using~\eqref{eq:kappa_bound_0} recursively implies
\begin{align}
\label{eq:kappa_bound_1}
\kappa_{T-1}(t+k+1)\leq& \kappa_{k}(t)-c_{\kappa}\sum_{j=0}^{k}\Delta \kappa(t+j)
\end{align}
for $k=0,\dots,T-1$. 
Using the definition of $\kappa$ in~\eqref{eq:kappa_def}, we can bound the $T$-step sum as 
\begin{align*}
&\sum_{k=0}^{T-1}\kappa_{T-1}(t+1+k)\stackrel{\eqref{eq:kappa_bound_1}}{\leq} \sum_{k=0}^{T-1}\left(\kappa_{k}(t)-c_{\kappa}\sum_{j=0}^k\Delta \kappa(t+j)\right)\\
\leq& \sum_{k=0}^{T-1}\kappa_k(t)-c_{\kappa}T\sum_{k=0}^{T-1}\Delta \kappa(t+k)\\
\stackrel{\eqref{eq:kappa_def}}{=}&\kappa(t)-c_{\kappa}T\sum_{k=0}^{T-1}\Delta \kappa(t+k)=\kappa(t)+c_{\kappa}T(\kappa(t)-\kappa(t+T)).
\end{align*}
Thus, the closed-loop transient cost over one period $T$ satisfies
\begin{align}
\label{eq:bound_1}
&W(t+T)-W(t)+ \sum_{k=t}^{t+T-1}\ell(x(k),u(k),k,y(k))\\
\leq&\kappa(t)+c_{\kappa}T(\kappa(t)-\kappa(t+T))+\sum_{k=0}^{T-1}\gamma(t+k)\kappa(t+1+k).\nonumber
\end{align}
 Abbreviate $\ell(t)=\ell(x(t),u(t),t,y(t))$ and define $\kappa_{\epsilon}(t)=\kappa(t)-\kappa_{\infty}$. 
Then~\eqref{eq:bound_1} evaluated over a time interval $K\cdot T$ starting at $t=0$ can be rewritten as
\begin{align}
\label{eq:bound_2}
&W(K\cdot T)-W(0)\\ 
\leq&K\kappa_{\infty}+c_{\kappa}T(\kappa(0)-\kappa(TK))+\sum_{k=0}^{K-1}\kappa_{\epsilon}(k\cdot T )\nonumber\\
&+\sum_{t=0}^{KT-1}[\gamma(t)\kappa_{\infty}+\gamma(t)\kappa_{\epsilon}(t+1)-\ell(t)].\nonumber
\end{align}
The remainder of the proof is analogous to~\cite[Thm.~1]{muller2013economic}. 
Boundedness of $\ell,~V_f$ and $\beta(t)\geq 0$ ensures that $W(TK)$ is lower bounded and thus
\begin{align*}
0\leq &\liminf_{K\rightarrow\infty} \dfrac{W(TK)-W(0)}{K}. 
\end{align*}
Taking averages on both sides of~\eqref{eq:bound_2} yields
\begin{align*}
0\leq &\liminf_{K\rightarrow\infty} \dfrac{W(TK)-W(0)}{K} \\
\leq&\kappa_{\infty}+\limsup_{K\rightarrow\infty}\dfrac{1}{K}\sum_{k=0}^{K-1}\kappa_{\epsilon}(k\cdot T )+\lim_{K\rightarrow\infty}\dfrac{c_{\kappa}T}{K}(\kappa(0)-\kappa(TK)) \\
&+\limsup_{K\rightarrow\infty}\dfrac{1}{K}\left[\sum_{t=0}^{KT-1}\gamma(t)\kappa_{\infty}+\gamma(t)\kappa_{\epsilon}(t+1)\right]\\
&-\limsup_{K\rightarrow\infty}\dfrac{1}{K}\sum_{t=0}^{KT-1}\ell(t)\\
\leq &\kappa_{\infty}-\limsup_{K\rightarrow\infty}\dfrac{1}{K}\sum_{t=0}^{TK-1}\ell(t),
\end{align*}
and thus~\eqref{eq:performance_bound1}. 
The first inequality follows from~\eqref{eq:bound_2}, by using 
\begin{align*}
\liminf_n a_n-b_n \leq \liminf_n -b_n+\limsup_n a_n=\limsup_n a_n -\limsup_n b_n. 
\end{align*}
The second inequality follows from $\gamma(t)\leq c_{\gamma}$, $\kappa_{\epsilon}(t)\in[0,\infty)$,  
\begin{align*}
\lim_{t\rightarrow\infty}\kappa_{\epsilon}(t)=0,~\limsup_{t\rightarrow\infty}\gamma(t)\leq 0,~\lim_{t\rightarrow\infty}\Delta \kappa(t)=0.
\end{align*}
\end{proof}

\subsection{Improved a priori performance bounds}
\label{sec:theory_2} 
In the following, we provide sufficient conditions to ensure that the average cost of the artificial periodic orbit converges to a local minimum. 
\subsubsection*{Terminal ingredients}
The following assumption is a stronger version of Assumption~\ref{ass:term_simple}, which is used to derive the improved performance guarantees. 
\begin{assumption}
\label{ass:term} 
Consider the terminal set $\mathcal{X}_f$, terminal cost $V_f$ and controller $k_f$ from Assumption~\ref{ass:term_simple}. 
There exists an incremental Lyapunov function $V_{\delta}(x,r_T,t)$, such that for any time $t\in\mathbb{N}$, any reference $r_T,\tilde{r}_T\in\mathcal{Z}_T(t)$ and any $x\in\mathcal{X}_f(r_T,t)$, the following inequalities hold
\begin{subequations}
\label{eq:term}
\begin{align}
\label{eq:term_increm}
V_{\delta}(x^+,r^+_T,t+1)-V_{\delta}(x,r_T,t)\leq &-\alpha_1(\|x-x_r\|),\\
\label{eq:term_bound}
\alpha_2(\|x-x_r\|)\leq  V_{\delta}(x,r_T,t)\leq& \alpha_3(\|x-x_r\|),\\
\label{eq:term_cont}
|V_{\delta}(x,r_T,t)-V_{\delta}(x,\tilde{r}_T,t)|\leq& \alpha_4(\|r_T-\tilde{r}_T\|),
\end{align}
\end{subequations}
with $x^+=f(x,u,t)$, $u=k_f(x,r_T,t)$, $r_T^+=\mathcal{R}_Tr_T\in\mathcal{Z}_T(t+1)$, $(x_r,u_r)=r_T(0)$ and functions $\alpha_1,~\alpha_2,~\alpha_3,~\alpha_4\in\mathcal{K}_{\infty}$.  
Furthermore, the terminal set is given by $\mathcal{X}_f(r_T,t)=\{x\in\mathbb{R}^n|~V_{\delta}(x,r_T,t)\leq \alpha(r_T)\}$ 
 and the terminal set size $\alpha(r_T)$ satisfies 
\begin{align}
\label{eq:term_alpha_cont}
|\alpha(r_T)-\alpha(\tilde{r_T})|\leq \alpha_5(\|r_T-\tilde{r_T}\|),~ \alpha(r_T)=\alpha(r_T^+)\in[ \underline{\alpha},\overline{\alpha}]
\end{align}
with constants $\underline{\alpha},\overline{\alpha}>0$ and a function $\alpha_5\in\mathcal{K}_{\infty}$. 
\end{assumption}
The offline design of such terminal ingredients is discussed in detail in Section~\ref{sec:term}. 
The conditions~\eqref{eq:term_increm}--\eqref{eq:term_cont} ensure that the terminal set has a non-empty interior and that the terminal controller $k_f$ stabilizes the reference $r_T$ with a continuous incremental Lyapunov function $V_{\delta}$. 
\begin{lemma}
\label{lemma:change_r}
Let Assumptions~\ref{ass:term_simple} and \ref{ass:term} hold. 
There exists a constant $\epsilon>0$, such that at each time $t\in\mathbb{N}$, for any $r_T\in\mathcal{Z}_T(t)$ and any $x\in\mathcal{X}_f(r_T,t)$, it holds that 
\begin{align*}
x^+=f(x,k_f(x,r_T,t),t)\in\mathcal{X}_f(\tilde{r}_T,t+1),
\end{align*}
for all $\tilde{r}_T\in\mathcal{Z}_T(t+1)\cap\mathbb{B}_{\epsilon}(\mathcal{R}_Tr_T)$. 
\end{lemma}
\begin{proof}
First, note that Assumption~\ref{ass:term} ensures that the positive invariance condition~\eqref{eq:term_PI} is strictly satisfied
\begin{align}
\label{eq:Delta_alpha}
&V_{\delta}(x^+,r_T^+,t+1)
\stackrel{\eqref{eq:term_increm}}{\leq} V_{\delta}(x,r_T,t)-\alpha_1(\|x-x_r\|)\nonumber\\
\stackrel{\eqref{eq:term_bound}}{\leq}& V_{\delta}(x,r_T,t)-\alpha_1(\alpha_3 ^{-1}(V_{\delta}(x,r_T,t)))\nonumber\\
\leq& \sup_{c\in[0,\alpha(r_T)]} c-\alpha_1(\alpha_3^{-1}(c))\nonumber\\
\stackrel{\eqref{eq:term_alpha_cont}}{\leq} & \alpha(r_T)-\underbrace{\min\{\alpha_1(\alpha_3^{-1}(\underline{\alpha}/2)),\underline{\alpha}/2\}}_{=:\Delta \alpha>0},
\end{align}
where the last step follows using the case distinction $c\leq \underline{\alpha}/2$ and $c\geq \underline{\alpha}$ and the fact that $\alpha(r_T)\geq \underline{\alpha}$. 
Given $\|r_T^+-\tilde{r}_T\|\leq \epsilon$, we have
\begin{align*}
V_{\delta}(x^+,\tilde{r}_T,t+1) 
\stackrel{\eqref{eq:term_cont}}{\leq}& V_{\delta}(x^+,r_T^+,t+1)+\alpha_4(\epsilon)\\
\stackrel{\eqref{eq:Delta_alpha}}{\leq}&\alpha(r_T)-\Delta \alpha+\alpha_4(\epsilon)\\
\stackrel{\eqref{eq:term_alpha_cont}}{\leq} &\alpha(\tilde{r}_T)+\alpha_5(\epsilon)+\alpha_4(\epsilon)-\Delta \alpha= \alpha(\tilde{r}_T),
\end{align*}
with $\epsilon:=(\alpha_4+\alpha_5)^{-1}(\Delta \alpha)$. 
\end{proof}
This lemma is an extension to~\cite[Lemma~1]{muller2014performance} and shows that the reference $r_T$ can be incrementally changed in closed-loop operation without losing recursive feasibility.
Similar results are used in nonlinear tracking MPC schemes~\cite{limon2018nonlinear,JK_periodic_automatica}.

\subsubsection*{Self-tuning weight}
Given a state $x$ at time $t$, the set of periodic reference trajectories $r_T$ with a terminal set $\mathcal{X}_f$ that can be reached within the prediction horizon $N$ is defined as
\begin{align*}
\mathcal{R}_N(x,t)&=\{r_T\in\mathcal{Z}_T(t+N)|~\exists u\in\mathbb{R}^{m N} \text{s.t. }x(t)=x, \\
&x(k+1)=f(x(k),u(k),k),~(x(k),u(k))\in\mathcal{Z}(k),\\
&k=t,\dots, t+N-1,\quad x(N+t)\in\mathcal{X}_f(r_T,t+N)\}. 
\end{align*}
Similarly, we define the set of reference trajectories that additionally satisfy the constraints~\eqref{eq:kappa_con_1}--\eqref{eq:kappa_con_2}
\begin{align*}
\overline{\mathcal{R}}_N(x,t,y,\kappa_j)&=\{r_T\in\mathcal{R}_N(x,t)|~\text{s.t.~} r_T \text{ satisfies }\eqref{eq:kappa_con_1}\text{--}\eqref{eq:kappa_con_2}\}. 
\end{align*}
Given a point $x\in\mathbb{R}^n$ at time $t$ with parameters $y$ and $\kappa_j$, the cost of the best reachable periodic orbit is given as
\begin{align}
\label{eq:J_T_min}
J_{T,\min}(x,t,y,\kappa_j)=\min_{r_T\in\overline{\mathcal{R}}_{N}(x,t,y,\kappa_j)} J_T(r_T,t+N,y).
\end{align}
\begin{assumption}
\label{ass:B1_2}
The update rule $\mathcal{B}$ is such that for any $y(t)=\overline{y}$ for all $t\geq0 $ and for all sequences $x(\cdot),~\kappa(\cdot)$, it holds that 
\begin{align*}
&\kappa_{\infty}-\liminf_{t\rightarrow\infty}J_{T,\min}(x(t),t,y(t),\kappa_j(t))>0\\
&\Rightarrow \liminf_{t\rightarrow\infty}\beta(t)=\infty.
\end{align*}
\end{assumption}
The main idea is that in closed-loop operation the self-tuning weight $\beta$ increases if necessary and thus ensures that the artificial trajectory converges to the optimal mode of operation, compare~\cite{muller2013economic,muller2014performance}. 
A detailed discussion on update schemes $\mathcal{B}$ satisfying Assumptions~\ref{ass:B1} and \ref{ass:B1_2} is given in~\cite{muller2013economic}. 

\subsubsection*{Periodic economic continuity}
As discussed in Section~\ref{sec:pitfal} and Section~\ref{sec:scheme}, the constraints~\eqref{eq:kappa_con_1}--\eqref{eq:kappa_con_2} are crucial for the desired properties. 
However, the constraint~\eqref{eq:kappa_con_2} limits how the shape of the artificial reference trajectory may change. 
In particular, for $c_{\kappa}=0$ this constraint ensures that the reference can only be updated if the economic cost on all points of the reference trajectory does not increase. 
For $c_{\kappa}$ arbitrarily large, the constraint~\eqref{eq:kappa_con_2} becomes inactive, if the overall cost of the artificial trajectory decreases ($\Delta \kappa<0$). 
However, both for numerical and technical reasons we consider the smooth constraint~\eqref{eq:kappa_con_2} with a finite value $c_{\kappa}$.   
Thus, we require the following technical continuity assumption on the periodic economic optimization problem~\eqref{eq:opt_periodic}. 
\begin{assumption}
\label{ass:cont_orbit}
There exists a positive constant $c_{\kappa}$, such that at any time step $t\in\mathbb{N}$, for any parameters $y\in\mathbb{Y}$, for any periodic trajectory $r_T\in\mathcal{Z}_T(t)$, which is not a local minimum of~\eqref{eq:opt_periodic} and any $\epsilon>0$, there exists a change $\Delta r_T$ with $\|\Delta r_T\|\leq \epsilon$, $r_T+\Delta r_T\in\mathcal{Z}_T(t)$, $J_T(r_T+\Delta r_T,t,y)<J_T(r_T,t,y)$, such that the following bound holds
\begin{align}
\label{eq:ass_c_kappa}
c_{\kappa}\geq \max_j\dfrac{\ell(r_T(j)+\Delta r_T(j),t+j,y)-\ell(r_T(j),t+j,y)}{J_T(r_T,t,y)-J_T(r_T+\Delta r_T,t,y)}.  
\end{align}
\end{assumption}	
This assumption ensures that it is possible to incrementally change the overall cost $J_{T}$, with incremental changes in the reference $r_T$ and the local cost $\ell$. 
If we expand the fraction by $\Delta r_T$ and take the limit $\Delta r_T\rightarrow 0$, we can see that this condition is similar to a continuity assumption on the fraction of the gradients of $\ell$ and $J_T$. 
Additional details regarding this condition are discussed in Sections~\ref{sec:term_cost_mod} and \ref{sec:convex}.  
The following lemma shows that this continuity condition (Ass.~\ref{ass:cont_orbit}) in combination with the incremental stabilizability property (Ass.~\ref{ass:term}) allows for the convergence to local minima. 
\begin{lemma}
\label{lemma:local_min}
Let Assumptions \ref{ass:term_simple}, \ref{ass:term} and ~\ref{ass:cont_orbit} hold and assume that $y(t)$ is constant. 
Suppose that the optimization problem~\eqref{eq:MPC} is feasible at time $t$ with some reference trajectory $r_T^*(\cdot|t)$, which is not a local minimum to~\eqref{eq:opt_periodic}. 
Then there exists a reference $\tilde{r}_T$ which is a feasible candidate solution to~\eqref{eq:MPC} at $t+1$, which satisfies 
\begin{align}
\label{eq:strict_improve}
J_T(\tilde{r}_T,t+N+1,y(t))< J_T(\mathcal{R}_Tr_T^*(\cdot|t),t+N+1,y(t)).
\end{align}
\end{lemma}
\begin{proof}
Given that $\mathcal{R}_Tr_T^*(\cdot|t)\in\mathcal{Z}_T(t+N+1)$ is not a local minimum, Assumption~\ref{ass:cont_orbit} ensures that there exists a reference $\tilde{r}_T=\mathcal{R}_Tr_T^*(\cdot|t)+\Delta r_T\in\mathcal{Z}_T(t+N+1)$, that improves the reference cost $J_T$~\eqref{eq:strict_improve} and satisfies $\|\Delta r_T\|\leq \epsilon$ and \eqref{eq:ass_c_kappa}. 
Satisfaction of the posed constraints~\eqref{eq:kappa_con_1}--\eqref{eq:kappa_con_2} follows from~\eqref{eq:ass_c_kappa}, by noting that $\Delta \kappa(t)=J_T(\tilde{r}_T,t+N+1,y(t))- J_T(\mathcal{R}_Tr_T^*(\cdot|t),t+N+1,y(t))$.
With $\epsilon$ according to Lemma~\ref{lemma:change_r}, the candidate input $u(\cdot|t+1)$ from Proposition~\ref{prop:av_perf} satisfies the terminal set constraint~\eqref{eq:MPC_term} with the incrementally changed reference $\tilde{r}_T$ and is thus a feasible solution to~\eqref{eq:MPC}. 
\end{proof}

\subsubsection*{A priori performance bounds}
The following proposition establishes a priori performance bounds on the artificial reference trajectory as an extension to~\cite[Thm~2/3, Corollary~1]{muller2014performance}. 
\begin{proposition}
\label{prop:art}
Let Assumptions \ref{ass:term_simple}, \ref{ass:term} and ~\ref{ass:cont_orbit} hold and assume that $y(t)$ is constant. 
Assume that the optimization problem~\eqref{eq:MPC} is feasible at $t=0$. 
If the update rule $\mathcal{B}$ satisfies Assumption~\ref{ass:B1_2}, then $\kappa_{\infty}$ is a local minimum of~\eqref{eq:opt_periodic}. 
\end{proposition}
\begin{proof} 
Using a proof of contradiction one can show $\kappa_{\infty}=\lim_{t\rightarrow\infty} J_{T,\min}(x(t),t,y(t),\kappa_j(t))$, compare~\cite[Thm.~2]{muller2014performance}, \cite[Thm.~2]{muller2013economic}. 
Furthermore, suppose that there exists a limiting reference $r_T$, which is not a local minimizer of~\eqref{eq:opt_periodic}. 
Lemma~\ref{lemma:local_min} ensures that there exists a feasible reference $\tilde{r}_T$, with an improved cost, which implies $J_{T,\min}<\kappa_{\infty}$ and thus contradicts the assumption.  
\end{proof}
The following theorem summarizes the theoretical properties of the proposed MPC scheme.   
\begin{theorem}
\label{thm:main}
Let Assumptions \ref{ass:term_simple}, \ref{ass:term} and ~\ref{ass:cont_orbit} hold and assume that~\eqref{eq:MPC} is feasible at $t=0$. 
Then the optimization problem~\eqref{eq:MPC} is recursively feasible for the resulting closed-loop system~\eqref{eq:close}.  
Assume further that $y(t)$ is constant and the update rule $\mathcal{B}$ satisfies Assumptions~\ref{ass:B1} and \ref{ass:B1_2}, then $\kappa_{\infty}$ is a local minimum of~\eqref{eq:opt_periodic} and the following performance bound holds   
\begin{align*}
\limsup_{K\rightarrow\infty}\dfrac{\sum_{t=0}^{TK-1}\ell(x(t),u(t),t,y(t))}{K}\leq \kappa_{\infty}.
\end{align*}
\end{theorem}
\begin{proof}
The results follow from Propositions~\ref{prop:feas}--\ref{prop:art}.   
\end{proof}

\begin{corollary}
\label{corol:performance}
Let Assumptions~\ref{ass:term_simple}, \ref{ass:term} and ~\ref{ass:cont_orbit} hold. 
Assume that~\eqref{eq:MPC} is feasible at $t=0$ and $y(t)$ is constant. 
If the update rule $\mathcal{B}$ is chosen as update scheme 2 or 6 in~\cite{muller2013economic}, then the closed-loop average economic performance is no worse than the performance at a locally optimal periodic orbit~\eqref{eq:opt_periodic}.  
\end{corollary}
\begin{proof}
This results follows directly from Theorem~\ref{thm:main}. 
It suffices to note, that the update schemes $2$ and $6$ satisfy Assumptions~\ref{ass:B1} and~\ref{ass:B1_2}, compare~ \cite[Lemmas~1 and 4]{muller2013economic}. 
\end{proof}
\begin{remark}
\label{rk:cont_time}
For simplicity, we have presented the proposed framework in a discrete-time setting.   
However, the approach can be directly applied to continuous-time problems  by defining the discrete-time stage cost $\ell$ and dynamics $f$ implicitly as the integration of some continuous-time dynamics $f_c$ and the average continuous-time cost $\ell_c$ over some sampling period $h$. 
One advantage of considering a continuous-time formulation is that the design of terminal ingredients satisfying Assumption~\ref{ass:term} (compare Section~\ref{sec:term}) simplifies, compare e.g.~\cite[App.~C]{JK_QINF}.
Furthermore, in a continuous-time setting it is possible to use a variable sampling time $h\in[h_{\min},h_{\max}]$, by considering the decision variable $u=(u_c,h)$, where $u_c$ denotes the (typically piece-wise constant) control input.
As a result, in a time-invariant setting the fixed constant $T$ does not directly impose a time length on the set of periodic orbits $\mathcal{Z}_T$, but only a finite parametrization. 
The constants $h_{\min}$, $h_{\max}$ need to be chosen, such that the (typically explicit) discretization scheme is stable and the MPC can react fast enough.
The advantages of such a formulation will also be explored in a numerical example in Section~\ref{sec:CSTR}. 
We point out that the benefits of using such a variable continuous-time period length have also been recently investigated in~\cite{gutekunst2020economic} using a direct multiple shooting method.
\end{remark}

\subsubsection*{Design parameters}
Overall, the proposed framework provides desired performance guarantees, if the constant $c_\kappa$, the self-tuning weight $\beta(t)$ and the terminal ingredients $V_f,\mathcal{X}_f$ are chosen properly (Ass.~\ref{ass:term_simple}--\ref{ass:cont_orbit}). 
In numerical experiments, we found that the closed loop is insensitive to changes in the constant $c_\kappa$, even by orders of magnitude, as long as $c_\kappa$ is sufficiently large (e.g. $c_\kappa=100$ in App.~\ref{app:CSTR}).
In Section~\ref{sec:term_cost_mod} we also show how the problem can be reformulated to get rid of the constant $c_\kappa$ and the continuity condition in Assumption~\ref{ass:cont_orbit}.
A large self-tuning weight $\beta(t)$ can deteriorate the transient performance, but is useful to ensure convergence of the artificial reference to a local minimum. 
In Section~\ref{sec:const_beta}, we show that similar performance bounds hold when choosing a constant weight $\beta$. 
For the special case of $T=1$ (artificial setpoint), more details on the effect of $\beta$ on the closed loop can be found in \cite{muller2013economic,muller2014performance} and \cite{fagiano2013generalized}.
Different design procedures for the terminal ingredients will be discussed in detail in Section~\ref{sec:term}.

\section{Proposed framework - details and variations}
\label{sec:ext}
In the following, we discuss details and variations  of the proposed framework. 
In Section~\ref{sec:term} we discuss in detail how to design terminal ingredients that satisfy Assumption~\ref{ass:term}.
In Section~\ref{sec:term_cost_mod} we discuss how to modify the cost function, such that the continuity condition in Assumption~\ref{ass:cont_orbit} can be dropped. 
In Section~\ref{sec:convex} we consider the special case of  convex periodic optimal control problems. 
In Section~\ref{sec:const_beta}, we discuss the theoretical properties without self-tuning weights $\beta(t)$, similar to~\cite{fagiano2013generalized}. 
%
\subsection{Terminal cost for economic dynamic operation}
\label{sec:term}
In the following, we detail how to design terminal ingredients that satisfy Assumption~\ref{ass:term}. 
First, in Section~\ref{sec:term_QINF} we show how a suitable tailored economic terminal cost $V_f$ can be designed using local linear and quadratic approximation of the dynamics $f$ and the economic cost $\ell$, as an extension and combination of the methods in~\cite{amrit2011economic,muller2014performance,JK_QINF}.
Then, in Section~\ref{sec:term_knowV} we show how a simple (and hence conservative) positive definite terminal cost $V_f$ can be computed based on any existing incremental Lyapunov function $V_{\delta}$, similar to the design in~\cite{alessandretti2016design,alessandretti2016convergence}.
Finally, in Section~\ref{sec:TEC}, we show that the theoretical properties can also be guaranteed with a simple terminal equality constraint (TEC), if a multi-step implementation is considered. 
%
\subsubsection{Reference generic offline computations}
\label{sec:term_QINF}
In the following, we detail a procedure to compute a suitable terminal cost $V_f$ (Ass.~\ref{ass:term}) based on the linearization of the dynamics and a quadratic approximations of the stage cost $\ell$ (using the Hessian and gradient).  
The following derivation is an extension of the approach in~\cite{amrit2011economic} to dynamic/periodic trajectories. 
Furthermore, we extend the approach to online optimized/changing reference trajectories by extending the reference generic offline computation from~\cite{JK_QINF} to an economic stage cost $\ell$.
In addition, this online computation involves an online computed adjoint for periodic trajectories, similar to the local gradient correction employed in~{\cite{zanon2018economic}.

\subsubsection*{Linear-quadratic local auxiliary stage cost}
%
The following Lemma extends the results in~\cite[Lemma~22-23]{amrit2011economic} to compute an auxiliary stage cost $\ell_q$, that locally upper bounds the stage cost $\ell$, which will later be used to derive sufficient conditions for inequality~\eqref{eq:term_dec}.} 
\begin{lemma}
\label{lemma:auxillary_cost}
Suppose there exists some $V_{\delta}$, $k_f$ satisfying the conditions in Assumption~\ref{ass:term}.
Suppose further that the sublevel sets of $V_{\delta}$ are convex in $x$, the controller $k_f$ is twice continuously differentiable in $x$, continuous in $r_T$,
 and satisfies  $k_f(x_r,r_T,t)=u_r$.
Suppose that  the stage cost $\ell$ and the dynamics $f$ are locally Lipschitz continuous and twice continuously differentiable w.r.t $(x,u)$. 
Then the function $\overline{\ell}(x,r_T,t,y)=\ell(x,k_f(x,r_T,t),t,y)-\ell(r,t,y)$ is twice continuously differentiable with respect to $x$. 
For any $\epsilon>0$, there exists a constant ${\alpha}_1>0$ and a  positive semi-definite matrix $S(r,t)\in\mathbb{R}^{n+m\times n+m}$, such that the following conditions hold for any $t\in\mathbb{N}$,  $r_T\in\mathcal{Z}_T(t)$, $y\in\mathbb{Y}$ and any $x\in\mathbb{R}^n$ with $V_{\delta}(x,r_T,t)\leq {\alpha}_1$:
\begin{align}
\label{eq:S_bound}
S(r,t)\succeq& \ell_{\xi\xi}(r,t,y),\\
\label{eq:aux_bound}
\ell_q(x,r_T,t,y)\geq& \overline{\ell}(x,r_T,t,y)+\frac{\epsilon}{2}\|x-x_r\|^2,
\end{align}
with 
\begin{align}
\label{eq:aux_stage}
\ell_q(x,r_T,t,y):=&\overline{\ell}_x(x_r,r_T,t,y) \cdot  (x-x_r)+\|x-x_r\|_{Q^*(r_T,t)}^2,\\
\label{eq:Q_star}
Q^*(r_T,t):=&
\begin{pmatrix}
I_n\\k_{f,x}(x_r,r_T,t)
\end{pmatrix}^\top 
S(r,t)
\begin{pmatrix}
I_n\\k_{f,x}(x_r,r_T,t)
\end{pmatrix}\nonumber\\
&+2\epsilon I_n+\sum_{j=1}^m \ell_{u_j}k_{f,j,xx}(x_r,r_T,t),
\end{align}
where $(x_r,u_r)=r=r_T(0)$, $\ell_{\xi\xi}$ denotes the Hessian of $\ell$ w.r.t. $\xi=(x,u)$, $k_{f,x}$ the Jacobian of $k_f$ w.r.t. $x$, $k_{f,j,xx}$ the Hessian of the $j$-th component of $k_f$ w.r.t. $x$ and $\ell_{u_j}$ the Jacobian of $\ell$ w.r.t. the $j$-th component of $u$. 
\end{lemma}

\begin{proof}
We point out that the derivative of $\overline{\ell}$ w.r.t. $x$ is the total derivative of $\ell$ w.r.t. $x$, for $u=k_f$. Hence,
the Jacobian and Hessian of $\overline{\ell}$ are given by
\begin{align*}
\overline{\ell}_x=&\ell_\xi
\begin{pmatrix}I_n\\k_{f,x}\end{pmatrix},\\
\overline{\ell}_{xx}=&\begin{pmatrix}I_n&k_{f,x}^\top\end{pmatrix}\ell_{\xi\xi}\begin{pmatrix}I_n&k_{f,x}^\top\end{pmatrix}^\top
+\sum_{j=1}^m \ell_{u_j}k_{f,j,xx},
\end{align*}
where $\ell_\xi\in\mathbb{R}^{1\times (n+m)}$, $\ell_{\xi\xi}\in\mathbb{R}^{(n+m)\times(n+m)}$ denote the Jacobian and Hessian of $\ell$ w.r.t. $\xi=(x,u)$.
Twice continuous differentiability of $\ell$ and compactness imply that there exists a finite constant 
\begin{align*}
c=\sup_{t\in\mathbb{N},r\in\mathcal{Z}_r(t),y\in\mathbb{Y}}\lambda_{\max}(\ell_{\xi\xi}(r,t,y)).
\end{align*}
Thus the matrix $S=(\max\{c,0\})I_{n+m}$ is positive semi-definite and satisfies $S\succeq \ell_{\xi\xi}$. 
The construction in~\eqref{eq:Q_star}, the definition of the Hessian $\overline{\ell}_{xx}$ and $S\succeq \ell_{\xi\xi}$ directly imply  $Q^*(r_T,t)\succeq \overline{\ell}_{xx}(x_r,r_T,t)+2\epsilon I_n$. 
Similar to~\cite[Lemma~22]{amrit2011economic}, there exists a small enough constant ${\alpha}_1>0$ (uniform in $r_T,t,y$), such that
$Q^*(r_T,t)\succeq \overline{\ell}_{xx}(x,r_T,t,y)+\epsilon I_n$, $\forall t\in\mathbb{N}$, $y\in \mathbb{Y}$, $r_T\in\mathcal{Z}_T(t)$, $V_{\delta}(x,r_T,t)\leq \alpha_1$. 
Abbreviate $\Delta x=x-x_r$, which implies $\ell_q=\|\Delta x\|_{Q^*}^2+\overline{\ell}_x \Delta x$. 
Convexity of the sublevel sets of $V_{\delta}$ implies that $V_{\delta}(x_r+s\Delta x,r_T,t)\leq \alpha_1$  for all $s\in[0,1]$ and any $V_{\delta}(x_r+\Delta x,r_T,t)\leq \alpha_1$. 
Hence, we can use the mean value theorem for vector functions~\cite[Prop.~A.11 (b)]{rawlings2017model}, similar to \cite[Lemma~23]{amrit2011economic}, to obtain
\begin{align}
\label{eq:eco_1}
&\ell_q(x,r_T,t,y)-\overline{\ell}(x,r_T,t,y)\\
=&\int_0^1 (1-s)\Delta x^\top (Q^*(r_T,t)-\overline{\ell}_{xx}(x_r+s\cdot \Delta x,r_T,t,y))\Delta x ds\nonumber\\
\geq& \int_0^1 (1-s)\epsilon\|\Delta x\|^2 ds=\epsilon/2  \|\Delta x\|^2.\nonumber
\end{align}
\end{proof}
Basically, $\ell_q$ is  a local linear-quadratic  over approximation of the stage cost $\ell$. 
Hence, we will formulate a sufficient condition for \eqref{eq:term_dec} using the auxiliary stage cost $\ell_q$. 
We point out that Lemma~\ref{lemma:auxillary_cost} does not impose any definiteness conditions on the Hessian of the stage cost $\ell$, but instead upper bounds the Hessian using the  positive semi-definite matrix $S$.

\subsubsection*{Sufficient conditions based on the linearization}
We denote the Jacobian of $f$ evaluated around an arbitrary point $r\in\mathcal{Z}_r(t)$ at some time $t\in\mathbb{N}$ by
\begin{align}
\label{eq:Jacob_dyn}
A(r,t):=\left.\left[\dfrac{\partial f}{\partial x}\right]\right|_{(x,u)=r},~ 
B(r,t):=\left.\left[\dfrac{\partial f}{\partial u}\right]\right|_{(x,u)=r}.
\end{align}
Given some periodic trajectory $r_T(\cdot|t)\in\mathcal{Z}_T(t)$, we denote the Jacobian w.r.t. $x$ of the system $f$ in closed loop with the terminal control law $k_f$ by 
\begin{align*}
&A_{cl}(r_T(\cdot|t),t)
:=A(r_T(0|t),t)+B(r_T(0|t),t)k_{f,x}(r_T,t).
\end{align*}
In the following, we introduce a corresponding adjoint periodic trajectory $p(k|t)$, which can be computed online based on the following set of $n\cdot T$ linear (in $p$) equality constraints
\begin{align}
\label{eq:compute_p}
&A_{cl}^\top(\mathcal{R}_T^j r_T(\cdot|t),t+j) p(j+1|t)\\
=&p(j|t)-\overline{\ell}_x^\top(x_{r_T}(j|t),\mathcal{R}_T^j r_T(\cdot|t),t+j,y),~j=0,\dots,T-1,\nonumber
\end{align}
with $p(N|t)=p(0|t)$. 
In the setpoint case ($T=1$), this reduces to $p^\top (A_{cl}-I)=-\overline{\ell}_x^\top$, similar to~\cite{amrit2011economic,muller2014performance}.
Similar to the adjoints used in \cite{zanon2018economic}, this vector $p$ corrects the effect of $\overline{\ell}_x$, the gradient of the stage cost.

The following proposition shows that such an online computed adjoint vector $p$ in combination with an offline computed matrix valued function $P$ provides a suitable terminal cost for dynamic operation with economic cost. 
\begin{proposition}
\label{prop:term_generic}
Suppose the conditions in Lemma~\ref{lemma:auxillary_cost} hold.
Assume further that there exists a positive definite $T$-periodic matrix $P(r_T,t)$, continuous in $r$, such that for any $t\in\mathbb{N}$, $r_T\in\mathcal{Z}_T(t)$, the following matrix inequality is satisfied
\begin{align}
\label{eq:P_ineq}
&A^\top_{cl}(r_T,t) P(\mathcal{R}_Tr,t+1)A_{cl}(r_T,t)-P(r_T,t) \\
\preceq& -Q^*(r_T,t)-\tilde{\epsilon}I_n, \nonumber
\end{align}
with $\tilde{\epsilon}>0$. 
Then  for any periodic reference $r_T\in\mathcal{Z}_T(t)$, $t\in\mathbb{N}$, the conditions~\eqref{eq:compute_p} have a unique solution $p(\cdot|t)$.
In addition, there exists a constant ${\alpha}_1$, such that the terminal cost 
\begin{align}
\label{eq:V_f}
V_f(x,r_T,t,y):=\|x_{r}-x\|_{P(r_T,t)}^2+p^\top(0|t) (x-x_{r}),
\end{align}
with $p$ according to~\eqref{eq:compute_p} and $x_r=x_{r_T}(0)$, satisfies condition~\eqref{eq:term_dec} with $\mathcal{X}_f(r_T,t):=\{x\in\mathbb{R}^n|~V_{\delta}(x,r_T,t)\leq \alpha_1\}$.
\end{proposition}
\begin{proof}
\textbf{Part I. } 
Condition~\eqref{eq:P_ineq} ensures that the linearized (time-varying) dynamics along the  periodic trajectory $r_T$ are (uniformly) exponentially stable, which implies
\begin{align}
\label{eq:A_det_1}
\det(I_n-\Pi_{j=0}^{T-1} A^\top_{cl}(\mathcal{R}_T^jr_T,t+j))>0.
\end{align}
Thus, the constraints~\eqref{eq:compute_p} have a unique solution $p$ for any $r_T\in\mathcal{Z}_T(t)$, compare also the reformulation of~\eqref{eq:compute_p} in Remark~\ref{rk:compute_p}.\\
\textbf{Part II. } 
Denote $\Delta x=x-x_r$. 
The first order Taylor approximation at $x=x_{r}$ yields
\begin{align*}
\Delta x^+=&f(x,k_{f}(x,r_T,t),t)-f(x_r,u_r,t)\\
=&A_{cl}(r_T,t)\Delta x+\Phi_{r_T,t}(\Delta x),
\end{align*}
with the remainder term $\Phi_{r_T,t}$. 
Twice continuous differentiability of $f$ and compact constraints imply that the remainder term is uniformly Lipschitz continuous in the terminal set, i.e., $\|\Phi_{r_T,t}(\Delta x)\|\leq L_{\Phi,{\alpha}_1}\|\Delta\|$ for all $t\geq 0$, with a constant $L_{\Phi,{\alpha}_1}$ arbitrary small for ${\alpha}_1$ arbitrary small. 
Using this bound in combination with condition~\eqref{eq:P_ineq} implies that there exists a sufficiently small constant ${\alpha}_1>0$, such that the nonlinear system (locally) satisfies
\begin{align}
\label{eq:term_quad}
&\|\Delta x^+\|_{P(\mathcal{R}_Tr_T,t+1)}^2-\|\Delta x\|_{P(r_T,t)}^2
\leq -\|\Delta x\|_{Q^*(r_T,t)}^2
\end{align}
for all $x\in\mathcal{X}_f(r_T,t)$, compare~\cite[Lemma~1]{JK_QINF} for details. \\
Given that $T$ is finite, $\overline{\ell}_x$ uniformly bounded and condition~\eqref{eq:A_det_1} holds, the vector $p$ admits a uniform bound.
Using the definition of $p$, we get
\begin{align}
\label{eq:term_lin}
&p^\top  (1)\Delta x^+\\
\leq& p^\top (1)A_{cl}(r_T,t)\Delta x+\|p(1)\|\|\Phi_{r_T,t}(\Delta x)\|\nonumber\\
\stackrel{\eqref{eq:compute_p}}{=}&p^\top(0) \Delta x-\overline{\ell} _x(x_{r},r_T,t,y)\Delta x+\|p(1)\|\|\Phi_{r_T,t}(\Delta x)\|\nonumber\\
\leq& p^\top(0) \Delta x-\overline{\ell} _x(x_{r},r_T,t,y)\Delta x+\epsilon/2\|\Delta x\|^2,\nonumber
\end{align}
where the last inequality holds for a sufficiently small constant ${\alpha}_1>0$, given the uniform bound on $\|p\|$ and the properties of the remainder term $\Phi_{r_T,t}$. 
By combining~\eqref{eq:term_quad} with \eqref{eq:term_lin} and using the auxiliary stage cost from Lemma~\ref{lemma:auxillary_cost}, the terminal cost~\eqref{eq:V_f} satisfies
\begin{align*}
&V_f(x^+,\mathcal{R}_T r_T,t+1,y)-V_f(x,r_T,t,y)\\
\leq& -\|\Delta x\|_{Q^*(r_T,t)}^2-\overline{\ell}_x(x_{r},r_T,t,y) \Delta x+\dfrac{\epsilon}{2}\|\Delta x\|^2\\
\stackrel{\eqref{eq:aux_bound}}{\leq}& -\overline{\ell}(x,r_T,t,y),
\end{align*}
and hence condition~\eqref{eq:term_dec}.
\end{proof}
\begin{corollary}
\label{corol:term_generic}
Suppose that  the stage cost $\ell$ and the dynamics $f$ are locally Lipschitz continuous and twice continuously differentiable w.r.t $(x,u)$. 
Assume further that there exists a positive definite matrix $P(r_T,t)$ and a matrix $k_{f,x}(r_T,t)$, both continuous in $r_T$ and T-periodic in $t$, such that for any $t\in\mathbb{N}$, $r_T\in\mathcal{Z}_T(t)$, the matrix inequality~\eqref{eq:P_ineq} is satisfied with some $\tilde{\epsilon}>0$. 
Then there exists a  function ${\alpha}(r_T)$, such that the terminal controller $k_f(x,r_T,t)=u_r+k_{f,x}(x-x_r)$, the terminal set $\mathcal{X}_f(r_T,t)=\{x\in\mathbb{R}^n|~\|x-x_r\|_{P(r_T,t)}^2\leq \alpha(r_T)\}$ and the terminal cost $V_f$ according to~\eqref{eq:V_f} satisfy Assumptions~\ref{ass:term_simple} and \ref{ass:term}.
\end{corollary}
\begin{proof}
Given that $P$ and $Q^*+\tilde{\epsilon}I_n$ are positive definite, the conditions~\eqref{eq:term_increm}--\eqref{eq:term_bound} in Assumption~\ref{ass:term} are satisfied with the incremental Lyapunov function $V_{\delta}(x,r_T,t)=\|x-x_r\|_{P(r_T,t)}^2$ and quadratic functions $\alpha_1,\alpha_2,\alpha_3\in\mathcal{K}_{\infty}$. 
Convexity of the terminal set $\mathcal{X}_f$ w.r.t. $x$ (compare conditions Lemma~\ref{lemma:auxillary_cost}) follows from $V_{\delta}$ quadratic in $x$.
Condition~\eqref{eq:term_cont} follows from $V_{\delta}$ quadratic and the assumed continuity of $P$ w.r.t. $r_T$.
Conditions~\eqref{eq:term_PI}, \eqref{eq:term_dec}, \eqref{eq:term} hold for any $\alpha\leq \alpha_1$, using Prop.~\ref{prop:term_generic}. 
Furthermore, given that $\mathcal{Z}_r(t)\in\text{int}(\mathcal{Z}(t))$ and $k_{f,x}$ bounded, there exists a small enough constant $\alpha_2>0$, such that conditions~\eqref{eq:term_con}, \eqref{eq:term_alpha_cont}  hold for any $\alpha\leq \alpha_2$. 
Hence, choosing the constant terminal set size $\alpha(r_T)=\min\{\alpha_1,\alpha_2\}$ satisfies all the conditions (condition \eqref{eq:term_alpha_cont} is trivially satisfied). 
\end{proof}

With this result, we can directly specify a procedure to compute suitable terminal ingredients.
First, a symbolic expression for the Jacobian $A$, $B$, $\ell_{\xi}$ and the Hessian $\ell_{\xi\xi}$ are computed. 
Then a positive semi-definite matrix $S$ is computed, which satisfies~\eqref{eq:S_bound}.
This can either be achieved with a constant matrix $S$ (c.f. proof Lemma~\ref{lemma:auxillary_cost}) or by computing a suitably parametrized matrix $S$ using linear matrix inequalities (LMIs).

Given $S$, we have to compute a parametrized matrix $P$, such that condition~\eqref{eq:P_ineq} holds. 
Suppose we want to compute a feedback of the form $k_f=u_r+k_{f,x}\Delta x$ with some parametrized feedback gain $k_{f,x}$ ($k_{f,xx}=0$).
In this case, condition~\eqref{eq:P_ineq} with $Q^*$ according to~\eqref{eq:Q_star} is equivalent to~\cite[Inequality~(36)]{JK_QINF}
with the following (output) tracking stage cost
\begin{align}
\label{eq:stage_output}
\tilde{\ell}=\|(C+Dk_{f,x})\Delta x\|_S^2+(2\epsilon+\tilde{\epsilon})\|\Delta x\|^2\\
C=\begin{pmatrix}I_n\\0_m\end{pmatrix}\in\mathbb{R}^{(n+m)\times n},~ D=\begin{pmatrix}0_n\\I_m\end{pmatrix}\in\mathbb{R}^{(n+m)\times m}.\nonumber
\end{align}
Hence, we can use the result in~\cite[Lemma~6, Prop.~6]{JK_QINF} to compute suitable matrices $k_{f,x}$ and $P$. 
In particular, in~\cite{JK_QINF} the matrices $P,k_{f,x},A,B,S$ are parametrized based on a quasi-LPV (linear-parameter-varying system) approach and the conditions are transformed into LMIs, that can be efficiently computed offline.

Given that the vector $p$ needs to satisfy condition~\eqref{eq:compute_p} with \textit{equality}, a similar parametrized offline computation for $p$ seems intractable (with the exception of linear systems, compare Section~\ref{sec:convex}).
Hence, we simply add\footnote{%
Due to the prediction horizon $N$ the time index $t$ changes to $t+N$ in~\eqref{eq:compute_p}.
}
the constraint~\eqref{eq:compute_p} to the MPC optimization problem~\eqref{eq:MPC} and compute $p(\cdot|t)$ online.

Finally, regarding the terminal set size $\alpha$, we first compute the constant $\alpha_1>0$, such that \eqref{eq:term_dec} holds for all $V_{\delta}\leq \alpha_1$, e.g. using Algorithm~1 from~\cite{JK_QINF}.
There are two options to compute a terminal set size $\alpha$ that also ensures constraint satisfaction~\eqref{eq:term_con}.
The definition of $\mathcal{Z}_r\subseteq\text{int}(\mathcal{Z})$ can be used to compute a constant $\alpha\in(0,\alpha_1]$, similar to the optimization problem~(24) in~\cite{JK_QINF}. 
However, such a constant $\alpha$ depends on the choice of $\mathcal{Z}_r$ and thus can yield arbitrary small values $\alpha$ (and thus slow convergence of $r_T$, compare Lemma~\ref{lemma:change_r}), or requires restrictive constraints on the set of periodic trajectories $\mathcal{Z}_T$.
This problem can be alleviated by computing a reference trajectory dependent terminal set size $\alpha(r_T)\in[\underline{\alpha},\overline{\alpha}]$ online, which can be done by using an additional scalar optimization variable $\alpha$ in \eqref{eq:MPC}, compare \cite[Sec.~3.3]{JK_periodic_automatica} for details.

The overall design procedure is summarized in Algorithm~\ref{alg:offline}.
\begin{algorithm}[H]
\caption{Offline computation}
\label{alg:offline}
\begin{algorithmic}[1]
\State  Compute Jacobian, Hessian $A$, $B$, $\ell_\xi$, $\ell_{\xi\xi}$. 
\State Determine matrix $S\succeq 0$, such that $S\succeq\ell_{\xi\xi}$~\eqref{eq:S_bound}. 
\State  Compute matrix $P$ (and possibly $k_f$) such that~\eqref{eq:P_ineq} holds using LMIs, compare~\cite[App.~D]{JK_QINF}.
\State  Compute maximal terminal set size $\alpha_1$ using~\cite[Alg.~1]{JK_QINF}. 
\State Derive $\alpha(r)\in(0,\alpha_1]$ for constraint satisfaction:
\Statex  ~a) Compute constant $\alpha_2>0$ using~\cite[Equation~(24)]{JK_QINF}\Statex~ and set $\alpha=\min\{\alpha_1,\alpha_2\}$.
\Statex  ~b) Compute $\alpha(r_T)$ online with a scalar optimization\Statex~variable, compare~\cite[Sec.~3.3]{JK_periodic_automatica}.
\end{algorithmic}
\end{algorithm}
The proposed procedure is a combination and extension of the reference generic offline computations in~\cite{JK_QINF}, the terminal ingredients for economic MPC~\cite{amrit2011economic,muller2014performance} and the online computation of $p$ using~\eqref{eq:compute_p}.  
Regarding the online operation, we simply include the constraints~\eqref{eq:compute_p} to compute $p(\cdot|t)$ and possibly constraints to compute $\alpha(r_T)$ online (compare \cite[Sec.~3.3.]{JK_periodic_automatica} for details) in the MPC optimization problem~\eqref{eq:MPC}. 
This procedure significantly simplifies in the special case of linear systems with linear/quadratic stage costs $\ell$, which is discussed in Section~\ref{sec:convex}. 
Furthermore, in the special case of artificial setpoints ($T=1$), we recover the schemes in~\cite{muller2013economic,muller2014performance,fagiano2013generalized} and Algorithm~\ref{alg:offline} provides a corresponding procedure to derive suitable terminal ingredients. 

 \begin{remark}
\label{rk:param_y}
The matrices $S$, $Q^*$ and $P$ can also be parametrized by $y$, to yield a terminal cost $V_f$ that depends on the price signal $y$. 
However, the incremental Lyapunov function $V_{\delta}$ used for the terminal set $\mathcal{X}_f$ (Ass.~\ref{ass:term}) may not depend on $y$ to ensure recursive feasibility independent of online changes in the price signal $y$. 
Thus, the choice of $\mathcal{X}_f$ in Corollary~\ref{corol:term_generic} is only valid for $P$ independent of $y$. 
\end{remark}

\begin{remark}
\label{rk:compute_p}
As already discussed, the vector $p$ needs to be computed online using~\eqref{eq:compute_p}, which adds $nT$ optimization variables and $n\cdot T$ equality constraints (linear in $p$) to the optimization problem~\eqref{eq:MPC}.
Abbreviate $A(j|t)=A(r_T(j|t),t+j)$, $B(j|t)=B(r_T(j|t),t+j)$, $\ell_\xi(j|t)=\ell_\xi(r_T(j|t),t+j,y(t))$ and suppose the feedback $k_{f,x}$ is parametrized in the form $k_{f,x}(j|t)=Y(j|t)X(j|t)^{-1}$ with matrices $X,Y$ (as is the case with the LMIs considered in~\cite{JK_QINF}). 
Then, multiplying the constraints~\eqref{eq:compute_p} by $X$ from the left yields the following equivalent constraint
\begin{align}
\label{eq:con_p_XY}
&(A(j|t)X(j|t)+B(j|t)Y(j|t))^\top p(j+1|t)\\
=&X(j|t)p(j|t)-\begin{pmatrix}
X(j|t)&Y^\top(j|t)
\end{pmatrix}\ell_\xi^\top(j|t),\nonumber
\end{align}
where we use $k_{f,x}X=Y$ and the formula for $\overline{\ell}_x$ from the proof of Lemma~\ref{lemma:auxillary_cost}.
The resulting constraint  can be implemented directly in terms of $X,Y$.
The constraints~\eqref{eq:compute_p} can also be compressed into one $n$ dimensional equality constraint with only the $n$ dimensional optimization variable $p(0|t)$.
In particular, denote $\overline{\ell}_x(j|t)=\overline{\ell}_x(x_{r_T}(j|t),r_T(\cdot|t),t+j,y(t))$ and $\overline{A}_{cl}(k|t)=\Pi_{j=0}^{k-1}A_{cl}(\mathcal{R}_T^jr_T(\cdot|t),t+j)$. 
Then $p(0|t)$ satisfying~\eqref{eq:compute_p} can be equivalently computed using
\begin{align}
\label{eq:direct_p}
(I_n-\overline{A}^\top_{cl}(T|t)) p(0|t)=\sum_{j=0}^{T-1}\overline{A}_{cl}^\top(j|t)\overline{\ell}^\top_x(j|t). 
\end{align}
Furthermore, if the period length $T$ is very large, an approximate solution can be obtained by assuming $\overline{A}_{cl}(j|t)\approx 0$ for $j\geq T_c$ with some $T_c<T$, which results in $p(0|t)\approx \sum_{j=0}^{T_c-1}\overline{A}_{cl}^\top(j|t)\overline{\ell}^\top_x(j|t)$. 
The fact that we need to take the full trajectory $r_T$ into account to compute the correct gradient correction $p$ indicates that the computation of a terminal cost for nonperiodic time-varying trajectories may be non-trivial. 
\end{remark}

\subsubsection{Given incremental Lyapunov function $V_{\delta}$}
\label{sec:term_knowV}
In the following, we consider the case, where some incremental Lyapunov function $V_{\delta}$ with a corresponding feedback $k_f$ is available and focus on computing a simple (positive definite) terminal cost $V_f$. 
For simplicity, we consider \textit{exponential} stability, as a special case of Ass.~\ref{ass:term}. 
\begin{assumption}
\label{ass:term_exp} 
There exists an incremental Lyapunov function $V_{\delta}(x,t,r_T)$, a controller $k_f(x,r_T,t)$, a terminal set size $\alpha(r_T)$ and constants $c_l,c_u>0$, $\rho\in(0,1)$, such that at any time $t\in\mathbb{N}$, for any reference $r_T\in\mathcal{Z}_T(t)$ and any $x\in\mathcal{X}_f(r_T,t):=\{x\in\mathbb{R}^n|~V_{\delta}(x,r_T,t)\leq \alpha(r_T)\}$, the following conditions hold
\begin{subequations}
\label{eq:term_2}
\begin{align}
\label{eq:term_increm_2}
V_{\delta}(x^+,r_T^+,t+1)\leq& \rho V_{\delta}(x,r_T,t),\\
\label{eq:term_bound_2}
c_l\|x-x_r\|\leq V_{\delta}(x,r_T,t)\leq& c_u\|x-x_r\|,\\
(x,u)\in&\mathcal{Z}(t),
\end{align}
\end{subequations}
with $x^+=f(x,u,t)$, $r_T^+=\mathcal{R}_Tr_T\in\mathcal{Z}_T(t+1)$, $u=k_f(x,r_T,t)$.  
In addition, Inequalities~\eqref{eq:term_cont} and \eqref{eq:term_alpha_cont} from Assumption~\ref{ass:term} hold with some $\alpha_4,\alpha_5\in\mathcal{K}_{\infty}$.  
\end{assumption}
Such an incremental Lyapunov function can for example be computed using quasi linear parameter varying methods~\cite{JK_QINF}, control contraction metrics~\cite{manchester2017control}, back stepping~\cite{zamani2013backstepping} or feedback linearization. 
In addition, we assume that the stage cost can be locally bounded by some polynomial, similar to~\cite[Ass.~5]{alessandretti2016design}, \cite[Ass.~26]{alessandretti2016convergence}. 
\begin{assumption}
\label{ass:stage_cost_bound}
Consider the incremental Lyapunov function $V_{\delta}$, the feedback $k_f$ and the terminal set $\mathcal{X}_f$ from Assumption~\ref{ass:term_exp}. 
There exists constants $a_k\geq 0$, $k=1,\dots,\nu$, such that the stage cost $\ell$ satisfies
\begin{align}
\label{eq:stage_cost_bound}
\ell(x,k_f(x,r_T,t),t,y)-\ell(r,t,y)\leq \sum_{k=1}^\nu a_k \|x-x_r\|^k,
\end{align}
for all $t\in\mathbb{N}$, $y\in\mathbb{Y}$, $r_T\in\mathcal{Z}_T(t)$, $x\in\mathcal{X}_f(r_T,t)$. 
\end{assumption}
The following proposition follows the arguments from~\cite[Prop.~2]{alessandretti2016design}, compare also~\cite[Prop.~27]{alessandretti2016convergence}. 
\begin{proposition}
\label{prop:Term_alessandretti}
Let Assumptions~\ref{ass:term_exp}--\ref{ass:stage_cost_bound} hold. 
Then the following terminal cost $V_f$ satisfies Assumptions~\ref{ass:term_simple} and \ref{ass:term}
\begin{align}
\label{eq:term_cost_alessandretti}
V_f(x,r_T,t):=\sum_{k=1}^{\nu} \dfrac{a_k}{c_l^{k}(1-\rho^{k})}V_{\delta}^{k}(x,r_T,t).
\end{align}
\end{proposition}
\begin{proof}
The terminal set $\mathcal{X}_f$ directly satisfies conditions~\eqref{eq:term_PI}, \eqref{eq:term_con} and Assumption~\ref{ass:term}. 
Condition~\eqref{eq:term_dec} follows with
\begin{align*}
&V_f(x^+,r_T^+,t+1)-V_f(x,r_T,t)\\
\stackrel{\eqref{eq:term_cost_alessandretti}}{=}& \sum_{k=1}^\nu \dfrac{a_k}{c_l^{k}(1-\rho^{k})}(V_{\delta}^{k}(x^+,r_T^+,t+1)-V_{\delta}^{k}(x,r_T,t))\\
\stackrel{\eqref{eq:term_increm_2}}{\leq}& \sum_{k=1}^\nu \dfrac{a_k}{c_l^{k}(1-\rho^{k})}(\rho^{k}-1)V_{\delta}^{k}(x,r_T,t)\\
=& -\sum_{k=1}^\nu a_k\left(\dfrac{V_{\delta}(x,r_T,t)}{c_l}\right)^{k}\stackrel{\eqref{eq:term_bound_2}}{\leq}-\sum_{l=1}^{\nu} a_k\|x-x_r\|^k\\
\stackrel{\eqref{eq:stage_cost_bound}}{\leq}& \ell(r,t,y)-\ell(x,k_f(x,r_T,t),t,y).
\end{align*}
\end{proof}
This terminal cost is easy to compute, but can also be quite conservative.  
In particular, an interesting feature of this approach is that the terminal cost $V_f$ is a \textit{sum} of incremental Lyapunov functions and as such also an incremental Lyapunov function, which is positive definite w.r.t. to the reference trajectory $r_T$. 
Correspondingly, the terminal cost incentives regulation towards the reference trajectory $r_T$.
In general the reference $r_T$ will have a suboptimal performance, which is why a purely economic (\textit{not} positive definite) terminal cost as in Proposition~\ref{prop:term_generic} may be advantageous. 

\begin{remark}
\label{rk:combine_linear_quadratic_pdf}
Note, that if a matrix $P$ is computed offline that satisfies condition \eqref{eq:P_ineq}, the results in Corollary~\ref{corol:term_generic} and Prop.~\ref{prop:Term_alessandretti} can be combined to show that the terminal cost $V_f=\|\Delta x\|_P^2+c\|\Delta x\|_P$ locally satisfies condition~\eqref{eq:term_dec}, by choosing $c>0$ suitably. 
In particular, consider $V_{\delta}=\|\Delta x\|_P$ satisfying Assumption~\ref{ass:term_exp} and the following local linear quadratic bound $\overline{\ell}\leq a_1\|\Delta x\|_P+\|\Delta x\|_{Q^*}^2$, with $Q^*$ according to Lemma~\ref{lemma:auxillary_cost} and a suitable constant $a_1>0$ using continuous differentiability of $\overline{\ell}$. 
Then choosing $c:=a_1/(1-\rho)$ ensures satisfaction of \eqref{eq:term_dec} using
\begin{align*}
&V_f^+-V_f\stackrel{\eqref{eq:term_quad},\eqref{eq:term_increm_2}}{\leq}
-\|\Delta x\|_{Q^*}^2-c(1-\rho)\|\Delta x\|_P 
\leq -\overline{\ell}.
\end{align*}
\end{remark}

\subsubsection{Terminal equality constraints} 
\label{sec:TEC}
In the following, we discuss how to replace the general terminal set (Assumption~\ref{ass:term}) with a simple terminal equality constraint (TEC, $\mathcal{X}_f(r_T,t)=x_r$). 
In principle, the conditions~\eqref{eq:term_increm}--\eqref{eq:term_cont} are quite general and not restrictive, but the explicit knowledge of $V_{\delta}$ (which characterizes the terminal set $\mathcal{X}_f$) can pose challenges. 

The following analysis is similar to~\cite{fagiano2013generalized} and \cite{limon2018nonlinear}, which also considered TEC in the steady-state case. 
To this end, we consider the following finite-time local incremental controllability condition. 
\begin{assumption}
\label{ass:TEC}
There exist constants $\nu\in\mathbb{N}$, $\epsilon>0$, such that at any time $t\in\mathbb{N}$, for any references $r_T,~\tilde{r}_T\in\mathcal{Z}_T(t)$ with $\|r_T-\tilde{r}_T\|\leq \epsilon$, there exists a state and input sequence $(x,u)\in\mathbb{R}^{\nu\times (n+m)}$, such that 
\begin{align*}
x(0)=&x_{r_T}(0),~ x(k+1)=f(x(k),u(k),t+k),\\
x(\nu)=&\tilde{x}_{r_T}(\nu),~ (x(k),u(k))\in\mathcal{Z}(t+k), ~ k=0,\dots,\nu-1.
\end{align*}
\end{assumption}
This condition is for example satisfied with $\nu\leq n$ if the linearization along any feasible periodic trajectory is (uniformly) controllable~\cite[Ass.~2.37]{rawlings2017model}, \cite[Ass.~2]{limon2014single}, compare also~\cite[Ass.~7]{fagiano2013generalized} \cite[Ass.~10]{muller2016economic}. 
Typically, an additional continuity bound on $\ell$ is used (compare for example~\cite[Ass.~4]{limon2018nonlinear}), 
which is, however, not necessary in the considered setup with the bounded stage cost $\ell$ and the self-tuning weight $\beta$. 
The following result is an adaptation of Lemmas~\ref{lemma:change_r}--\ref{lemma:local_min} to terminal equality constraints (TEC) using Assumption~\ref{ass:TEC} and a multi-step implementation.
\begin{lemma}
\label{lemma:change_r_TEC}
Let Assumptions~\ref{ass:cont_orbit} and \ref{ass:TEC} hold  and assume that $y(t)$ is constant. 
Consider the terminal ingredients $\mathcal{X}_f(r_T,t)=\{x_r\}$, $V_f(x,r_T,t,y)=0$ and a prediction horizon $N\geq \nu$.  
Suppose that the optimization problem~\eqref{eq:MPC} is feasible at time $t$ with some reference trajectory $r_T^*(\cdot|t)$, which is not a local minimum to~\eqref{eq:opt_periodic}. 
Then under the $\nu$-step closed-loop system
\begin{align}
\label{eq:TEC_multistep}
&u(t+k)=u^*(k|t),\quad x(t+k+1)=x^*(k+1|t),\\
\quad &k=0,\dots, \nu-1,\nonumber
\end{align}
 there exists a reference $\tilde{r}_T$ which is a feasible candidate solution to~\eqref{eq:MPC} at $t+\nu$ and satisfies 
\begin{align*}
J_T(\tilde{r}_T,t+N+\nu,y(t))< J_T(\mathcal{R}_T^{\nu}r_T^*(\cdot|t),t+N+\nu,y(t)).
\end{align*}
\end{lemma}
\begin{proof}
Given that $\mathcal{R}_T^{\nu}r_T^*(\cdot|t)\in\mathcal{Z}_T(t+N+\nu)$ is not a local minimum, Assumption~\ref{ass:cont_orbit} ensures that there exists a reference $\tilde{r}_T=\mathcal{R}_T^{\nu}r_T^*(\cdot|t)+\Delta r_T\in\mathcal{Z}_T(t+N+\nu)$, that satisfies the posed constraints~\eqref{eq:kappa_con_1}--\eqref{eq:kappa_con_2}, improves the reference cost $J_T$ and satisfies $\|\tilde{r}_T-\mathcal{R}_T^{\nu}r_T^*(\cdot|t)\|\leq \epsilon$. 
Due to the multi-step implementation~\eqref{eq:TEC_multistep}, the sequence $x(k|t+\nu)=x^*(k+\nu|t)$,~$u(k|t+\nu)=u^*(k+\nu|t)$ satisfies the dynamics~\eqref{eq:MPC_dyncon}, the constraints~\eqref{eq:MPC_cons} and ends in the reference~\eqref{eq:MPC_term}, i.e., 
\begin{align*}
x(N-\nu|t+\nu)=x^*(N|t)=x_{r_T}^*(0|t).
\end{align*} 
Correspondingly, for any $\tilde{r}_T\in\mathcal{Z}_T(t+N+\nu)$ with $\|\mathcal{R}_T^{\nu}r_T^*(\cdot|t)-\tilde{r}_T\|\leq \epsilon$, we can append the state and input sequence $x(\cdot|t+\nu)$, $u(\cdot|t+\nu)$ with the candidate solution from Assumption~\ref{ass:TEC}, which satisfies the constraints~\eqref{eq:MPC_cons}--\eqref{eq:MPC_periodic}. 
\end{proof}
Compared to the results in Lemmas~\ref{lemma:change_r}--\ref{lemma:local_min} based on terminal sets, the resulting properties with terminal equality constraints are only valid if we apply the first $\nu$ parts of the computed input sequence. 
For comparison, in tracking MPC with positive definite stage cost $\ell$, such a multi-step implementation is not needed, since the closed-loop system eventually converges to an $\epsilon$ neighbourhood of the reference trajectory $r_T$, compare~\cite[Thm.~2]{limon2018nonlinear}, \cite[Thm.~3]{limon2014single} \cite[Thm.~2]{limon2016mpc}.

The main benefit of the terminal equality constraint implementation is the simple design. 
Although we need to use a multi-step implementation with $\nu$ steps, we would like to point out that $\nu$ is independent of $T$, and hence this method does not suffer from the same limitations as the approaches based on $T$-step systems, such as~\cite{muller2016economic,wabersich2018economic}. 
On the other hand, an implementation with a suitable terminal cost (Ass.~\ref{ass:term}) can use any prediction horizon $N$, requires no multi-step implementation and typically yields better closed-loop performance and (inherent) robustness properties, compare the numerical example in Section~\ref{sec:CSTR}.

\begin{remark}
In~\cite{fagiano2013generalized} a similar economic MPC scheme for setpoints ($T=1$) has been considered with terminal equality constraints. 
However, instead of a $\nu$-step MPC implementation~\eqref{eq:TEC_multistep}, in~\cite[Algorithm~3]{fagiano2013generalized} it was suggested to augment the MPC with an algorithm that decides at each time $t$ if the candidate solution or the standard MPC feedback~\eqref{eq:close_1} is applied. 
In particular, if the cost of the artificial reference $r_T$ does not improve by a minimal amount $\tilde{\epsilon}$, the candidate solution is applied. 
Given Assumption~\ref{ass:TEC}, after at most $\nu$ steps, it is possible to incrementally move the reference trajectory and thus improve the cost. 
Thus, by augmenting the MPC with such an algorithm, it may not be necessary to apply the first $\nu$ steps of the computed input trajectory, which can speed up convergence. 
\end{remark}

\subsection{Modified reference cost}
\label{sec:term_cost_mod}
The proposed formulation~\eqref{eq:MPC} uses standard conditions for the terminal ingredients (Ass.~\ref{ass:term_simple}) and contains many economic MPC formulations as special cases, compare~\cite{muller2013economic,muller2014performance,fagiano2013generalized,houska2017cost,wang2018economic,zanon2017periodic,alessandretti2016convergence,angeli2012average,amrit2011economic}. However, the formulation also requires the additional constraints~\eqref{eq:kappa_con_1}--\eqref{eq:kappa_con_2}, based on the continuity condition (Ass.~\ref{ass:cont_orbit}).
In the following, we briefly discuss an alternative solution to this problem, based on a modified cost for the artificial reference trajectory.
The following result is based on~\cite[Prop.~1]{kohler2018periodic}, which in turn is motivated by the analysis of non-monotonic Lyapunov functions~\cite{ahmadi2008non}.
\begin{lemma}
\label{lemma:term_cost_mod}
Consider the terminal ingredients from Assumption~\ref{ass:term_simple}.
For any $t\in\mathbb{N}$, $y\in\mathbb{Y}$, $r_T\in\mathcal{Z}_T(t)$, $x\in\mathcal{X}_f(r_T,t)$, 
the modified terminal cost
\begin{align*}
\tilde{V}_f(x,r_T,t,y):=&V_f(x,r_T,t,y)\\
&+\sum_{k=0}^{T-2}\frac{T-1-k}{T}\ell(r_T(k),t+k,y)
\end{align*}
satisfies 
\begin{align*}
&\tilde{V}_f(x^+,\mathcal{R}_Tr_T,t+1,y)-\tilde{V}_f(x,r_T,t,y)\\
\leq& -\ell(x,u,t,y)+J_T(r_T,t,y)/T,
\end{align*}
with $x^+=f(x,u,t)$, $u=k_f(x,r_T,t)$.
\end{lemma}
\begin{proof}
Abbreviate $\ell(k)=\ell(r_T(k),t+k,y)$, $\tilde{V}_f=\tilde{V}_f(x,r_T,t,y)$, $\tilde{V}_f^+=\tilde{V}_f(x^+,\mathcal{R}_Tr_T,t+1,y)$.
The modified terminal cost satisfies
\begin{align*}
&T(\tilde{V}_f^+-\tilde{V}_f)\\
=&T(V_f^+-V_f)+\sum_{k=0}^{T-2} (T-1-k)(\ell(k+1)-\ell(k))\\
\stackrel{\eqref{eq:term_dec}}{\leq}&-T\ell(x,u,t,y)+(T+1-T)\ell(0)+\ell(T-1)+\sum_{k=1}^{T-2}\ell(k)\\
=&-T\ell(x,u,t,y)+J_T(r_T,t,y).
\end{align*}
\end{proof}
We point out, that the modification of the cost in Lemma~\ref{lemma:term_cost_mod} is applicable both to terminal equality constraints (Lemma~\ref{lemma:change_r_TEC}) and terminal cost/sets (Ass.~\ref{ass:term}). 
The following proposition shows that this modified terminal cost can ensure the same theoretical properties (Prop.~\ref{prop:av_perf}) as the proposed scheme~\eqref{eq:MPC}, without using the continuity condition (Ass.~\ref{ass:cont_orbit}).
\begin{proposition}
\label{prop:av_perf_2}
Let Assumption~\ref{ass:term_simple} hold and assume that $y(t)$ is constant.  
Consider the MPC formulation~\eqref{eq:MPC} with $V_f$ replaced by $\tilde{V}_f$ (Lemma~\ref{lemma:term_cost_mod}) and the constraints~\eqref{eq:kappa_con_1}--\eqref{eq:kappa_con_2} replaced by 
\begin{align}
\label{eq:kappa_con_new}
 J_{T}(r_T(\cdot|t),t+N,y(t))\leq \kappa(t):=\sum_{j=0}^{T-1}\kappa_j(t). 
\end{align}
The resulting closed-loop system satisfies the performance bound \eqref{eq:performance_bound1}, if the update rule $\mathcal{B}$ satisfies Assumption~\ref{ass:B1}.  
\end{proposition}
\begin{proof}
Similar to Prop.~\ref{prop:av_perf} the candidate solution from Prop.~\ref{prop:feas} with the modified terminal cost implies
\begin{align*}
&W(t+1)-W(t)+\ell(x(t),u(t),t,y(t))\\
\leq& J_T(r_T^*(\cdot|t),t+N,y(t))/T+\gamma(t)\kappa(t+1)\nonumber\\
\stackrel{\eqref{eq:kappa_def}}{=} &\kappa(t+1)/T+\gamma(t)\kappa(t+1).\nonumber
\end{align*}
Correspondingly, the $T$-step bound~\eqref{eq:bound_1} holds with $c_{\kappa}=0$, since $\kappa(t+1)\leq \kappa(t)$. 
The remainder of the proof follows from the arguments in Prop.~\ref{prop:av_perf}, similar to~\cite[Thm.~1]{muller2013economic}\cite[Thm.~1]{muller2014performance}. 
\end{proof}
The properties in Prop.~\ref{prop:art} and Theorem~\ref{thm:main} hold equally with the modified terminal cost $\tilde{V}_f$ and a terminal set (Ass.~\ref{ass:term}), with the simpler constraint~\eqref{eq:kappa_con_new} (without requiring Assumption~\ref{ass:cont_orbit}). 
The main advantage of using this modified terminal cost $\tilde{V}_f$ is that the technical continuity condition Assumption~\ref{ass:cont_orbit} is not required and the number of constraints in \eqref{eq:MPC} is smaller, while the theoretical properties are the same. 
However, this modified terminal cost yields an objective function, which to the best knowledge of the authors differs from any existing MPC formulation (for $T>1$). 
Correspondingly, it is unclear what the practical effect on the closed-loop performance is, which will be studied in the numerical example in Section~\ref{sec:CSTR}. 
For the example considered in Section~\ref{sec:pitfal} with a terminal equality constraint and $T=2$, the modified terminal cost is $\tilde{V}_f=\frac{1}{2}\ell(r)$. 
With this modified cost the closed loop also \textit{does the right thing}, i.e., converges to the $T$-periodic orbit $\{1,2\}$. 
%
\subsection{Convex problem}
\label{sec:convex}
In the following, we discuss the special case, when the periodic optimal problem~\eqref{eq:opt_periodic} is convex. 
Suppose that the dynamics $f$ are affine, i.e. $f(x,u,t)=A(t)x+B(t)u+c(t)$, and the constraint sets $\mathcal{Z}_r$ are polytopes, which implies that  $\mathcal{Z}_T$ is a convex polytope. 
For $\ell$ convex, this implies that the  periodic optimal problem~\eqref{eq:opt_periodic} is convex and Theorem~\ref{thm:main}/Corollary~\ref{corol:performance} guarantee that the closed-loop performance is no worse than operation at an optimal $T$-periodic orbit. 

In Section~\ref{sec:convex_term} we discuss how the construction of the terminal ingredients (Sec.~\ref{sec:term}) and the online optimization simplifies in the convex case.
In Section~\ref{sec:convex_continuity}, we discuss the continuity condition (Ass.~\ref{ass:cont_orbit}) for periodic optimal control. 
%
%
\subsubsection{Offline design and online optimization} 
\label{sec:convex_term}
In the following, we discuss how the design procedure (Sec.~\ref{sec:term}) and the online optimization simplifies for the considered special case. 

Since we have a linear (time-varying) system, we consider a linear time-varying feedback $k_{f,x}=K(t)$ and a time-varying matrix $S(t)$ satisfying condition~\eqref{eq:S_bound}.  
Thus, the matrix $Q^*(t)$ in Lemma~\ref{lemma:auxillary_cost} is independent of $r_T$ and we can consider a time-varying matrix $P(t)$ to satisfy 
condition~\eqref{eq:P_ineq} in Proposition~\ref{prop:term_generic}. 
Matrices $P(t)$, $K(t)$ satisfying condition~\eqref{eq:P_ineq} can be computed by solving $T$ coupled LMIs similar to~\cite{aydiner2016periodic}. 
Alternatively, the computation of $K(t)$, $P(t)$ can be achieved using the discrete-time LQR for a suitably defined $T$-step system with $\tilde{x}\in\mathbb{R}^n$ and $\tilde{u}\in\mathbb{R}^{Tm}$.

Using the reformulation of the constraints~\eqref{eq:compute_p} discussed in Remark~\ref{rk:compute_p}, the possibly nonlinear constraints~\eqref{eq:compute_p} can be dropped by adding the explicit nonlinear term for $p(0|t)$ in the cost function. 
For the terminal set $\mathcal{X}_f$, we can either use an ellipsoidal set $\mathcal{X}_f(r_T,t)=\{x|~\|x-x_r\|_{P(t)}^2\leq \alpha\}$ or a polytopic (periodically time-varying) invariant set $\mathcal{X}_f$\footnote{%
The optional consideration of an online optimized terminal set size $\alpha(r_T)$ can be expressed using linear constraints, for both cases.
}. 
Thus, $\ell$ convex implies that  the constraints in~\eqref{eq:MPC} are convex.

In case $\ell$ quadratic, Lemma~\ref{lemma:auxillary_cost} and Prop.~\ref{prop:term_generic} contain no nonlinear terms that need to be locally over-approximated and hence  we can set $\epsilon=\tilde{\epsilon}=0$ and $\alpha_1$ arbitrary large.
%
Furthermore, for $\ell$ quadratic, the problem~\eqref{eq:MPC} is a quadratically constrained\footnote{%
 In addition to the possibly ellipsoidal terminal set, the constraints~\eqref{eq:kappa_con_2} are quadratic, leading to a non negligible increase in the online computation. 
Given that $p(0|t)$ in~\eqref{eq:direct_p} is linear in $r_T$,  the terminal cost $V_f$ is quadratic in the decision variables.} quadratic program (QCQP).
If we drop the constraints~\eqref{eq:kappa_con_2} (compare Section~\ref{sec:term_cost_mod}) and use a polytopic terminal set $\mathcal{X}_f$, the optimization problem~\eqref{eq:MPC} is a (linearly constrained) quadratic program (QP). 

In the special case that $\ell$ is linear, i.e., $\ell=x^\top q(t,y)$, the vector $p$ can be explicitly computed (independent of the online optimized reference $r_T$) for a given price $y$.  
Furthermore, since $S=0$, the terminal cost is linear ($P=0$). 
Thus, if a polyhedral terminal set is chosen, the proposed scheme with a linear cost $\ell$ only requires the solution to a linear program (LP), which can be done efficiently. 
Furthermore, with a polytopic incremental Lyapunov function $V_{\delta}$, the terminal cost in Proposition~\ref{prop:Term_alessandretti} can be formulated in the MPC problem~\eqref{eq:MPC} using linear constraints, resulting in an LP (or QP in case of $\ell$ quadratic). 

For this special case with a prediction horizon of $N=0$, the proposed MPC scheme is almost equivalent to the periodicity constraint MPC proposed in~\cite{houska2017cost,wang2018economic}.
The main difference is that we use a linear terminal cost and the polyhedral constraint, instead of a terminal equality constraint (and the constraints~\eqref{eq:kappa_con_2} which can typically be neglected). 
This small difference in the design allows us to derive the desired performance guarantees, while the periodicity constraint MPC can lead to suboptimal performance, compare \cite[Example 6]{wang2018economic}.
We can get the same theoretical properties with a terminal equality constraint (TEC) and $N\geq \nu$, if we use Lemma~\ref{lemma:change_r_TEC}.

In case that some of the input variables $u$ are also subject to integer constraints (c.f. for example periodic scheduling problems with discrete decisions~\cite{risbeck2019unification} and the HVAC numerical example in Section~\ref{sec:HVAC}), the problem can be formulated as a mixed-integer linear program  (MILP).

\subsubsection{Continuity condition - Assumption~\ref{ass:cont_orbit}}
\label{sec:convex_continuity}
In the following, we discuss sufficient conditions for Assumption~\ref{ass:cont_orbit}. 
Suppose that $\ell$ is continuously differentiable and $\mathcal{Z}_T$ is a convex polytope. 
Given a reference $r_T$, the direction of feasible changes $\Delta r_T$, which imply a decrease in $J_T$,  i.e. $\{\Delta r_T^\top \nabla J_T\leq 0,~r_T+\Delta r_T\in \mathcal{Z}_T\}$, is  a polytope. 
The condition in Assumption~\ref{ass:cont_orbit} reduces to the existence of a direction $\Delta r_T$ in this set, such that the directional derivative of $\ell$ is uniformly bounded relative to the directional derivative of $J_T$.

Although this condition is reasonable in the context of optimal periodic control, it is not necessarily satisfied for any convex problem. 
In particular, it is not valid if the (directional) derivative of $J_T$ vanishes, but the gradient of $\ell$ is (uniformly) lower bounded. 
This is for example the case, if the economic cost $J_T$ is quadratic and the inequality constraints are not active at the optimal periodic orbit, yielding a vanishing gradient of $J_T$.
If this problem occurs, we need to use the reformulation in Section~\ref{sec:term_cost_mod} to guarantee optimal performance without Assumption~\ref{ass:cont_orbit}.

Suppose that there exists a constant $\epsilon>0$, such that for any $r_T\in\mathcal{Z}_T$ which is not a local minimum, there exists a feasible direction $\Delta r_T$ with $\nabla J_T^\top \Delta r_T\leq -\epsilon \|\Delta r_T\|$. 
Then Assumption~\ref{ass:cont_orbit} is satisfied with some finite $c_{\kappa}$, if $\ell$ is Lipschitz continuous. 
Note that such a directional  derivative always exists if, e.g.,  $\ell$ is linear.  

\subsection{Constant parameter $\beta$}
\label{sec:const_beta}
In the following, we briefly discuss how the performance bounds in Proposition~\ref{prop:art} can be approximately guaranteed with a constant weight $\beta$. 
In particular, in~\cite{fagiano2013generalized} a competing approach to~\cite{muller2013economic,muller2014performance} has been considered with a constant weight $\beta$. 
Instead of changing the weight $\beta$ online to achieve (locally) \textit{optimal} performance, a fixed weight $\beta$ is considered and a suboptimality bound on the performance is established. 
The following proposition shows that the same result applies here, as an alternative to Proposition~\ref{prop:art}, similar to~\cite[Prop.~2]{fagiano2013generalized}. 
\begin{proposition}
Let Assumptions \ref{ass:term_simple} and ~\ref{ass:cont_orbit} hold with a bounded terminal cost $V_f$ and assume that $y(t)$ and $\beta(t)$ are constant. 
Assume that the optimization problem~\eqref{eq:MPC} is feasible at time $t$.
There exists a function $\underline{\beta}$, such that for any $\epsilon>0$, $\beta\geq \underline{\beta}(\epsilon)$ implies $\kappa(t+1)\leq J_{T,\min}(x(t),t,y,\kappa_j(t))+\epsilon$. 
\end{proposition}
\begin{proof}
Denote the minimizer and minimum of~\eqref{eq:J_T_min} at time $t$ by $r_{\min}(t)$ and $J_{T,\min}(t)$, respectively. 
By definition, there exists a feasible input sequence $\tilde{u}$ with corresponding state sequence $\tilde{x}$ such that $r_{\min}(t)$ satisfies the constraints in~\eqref{eq:MPC}. 
Due to optimality, the cost
\begin{align*}
\tilde{W}=&\sum_{k=0}^{N-1}\ell(\tilde{x}(k),\tilde{u}(k),t+k,y)\\
&+V_f(\tilde{x}(N),r_{\min}(t),t+N,y)
+\beta J_{T,\min}(t),
\end{align*}
satisfies $W(t)\leq\tilde{W}$, which is equivalent to
\begin{align*}
&\beta (J_T(r_T^*(\cdot|t),t+N,y)-J_{T,\min}(t))\\
\leq &
\sum_{k=0}^{N-1}\ell(\tilde{x}(k),\tilde{u}(k),t+k,y)- \ell(x^*(k|t),u^*(k|t),t+k,y)\\
&+V_f(\tilde{x}(N),r_{\min}(t),t+N,y)\\
&-V_f(x^*(N|t),r_T^*(\cdot|t),t+N,y)\\
\leq& \eta,
\end{align*}
with some constant $\eta>0$.
The last inequality follows from boundedness of $\ell$, $V_f$ and $N$ finite. 
This inequality directly implies
\begin{align*}
\kappa(t+1)=J_T(r_T^*(\cdot|t),t+N,y)\leq J_{T,\min}(t)+\epsilon,
\end{align*}
for $\beta\geq \underline{\beta}(\epsilon):=\eta/\epsilon$. 
\end{proof} 
Thus, for a large enough weight $\beta$, the cost of the artificial periodic orbit $r_T$ is arbitrarily ($\epsilon$) close to the cost of the optimal reachable periodic orbit~\eqref{eq:J_T_min}. 
Combining this result with Lemma~\ref{lemma:local_min} and the stronger terminal ingredients (Ass.~\ref{ass:term}), $J_{T,\min}$ is a local minimizer to~\eqref{eq:opt_periodic}.
Correspondingly, it is possible to derive performance bounds similar to Theorem~\ref{thm:main} with an additional suboptimality term $\epsilon$, compare~\cite[Thm.~2]{fagiano2013generalized}. 
%

\section{Numerical Examples}
\label{sec:num}
The following examples demonstrate the applicability of the proposed framework to  dynamic operation and online changing conditions. 
We first consider a simple HVAC systems~\cite{risbeck2019economic}, where dynamics, cost and constraints are periodically time-varying and discrete inputs are considered.
Then we consider the classical problem of increasing the yield of continuous stirred-tank reactors (CSTR) with dynamic operation, compare~\cite{bailey1971cyclic}. 
In this example, we compare the performance of the proposed approach with  periodicity constraint MPC~\cite{houska2017cost,wang2018economic},  tracking MPC formulations~\cite{limon2014single,JK_periodic_automatica} and MPC without terminal constraints~\cite{grune2013economic,muller2016economic,grune2018economic}.
In addition, we study the effect of the various degrees of freedoms in the formulation (terminal ingredients (Sec.~\ref{sec:term}), alternative cost formulations (Sec.~\ref{sec:term_cost_mod}), continuous-time formulations (Remark~\ref{rk:cont_time})) on closed-loop performance.
 
\subsection{Building cooling}
\label{sec:HVAC}
In the following, we show the applicability of the proposed framework to periodic optimal control problems subject to online changing performance measures. 
We consider a simple building temperature evolution example from~\cite[Sec.~IV.A]{risbeck2019economic} governed by
\begin{align*}
m\dfrac{d}{dt}T(t)=-k(T(t)-T_{\text{amb}}(t))+q_{\text{amb}}(t)-q(t),
\end{align*}
with air temperature $T$, cooling rate $q$, ambient temperature $T_{\text{amb}}$, rate of direct heat by the ambient $q_{\text{amb}}$ and  model constants $m,k>0$. 
The cooling rate $q$ is generated using $N_{\text{chiller}}=2$ chillers and is subject to the following (time-invariant) disjoint constraint set 
\begin{align*}
q\in\mathbb{U}=\{0\}\times [0.75,1]\times [1.5,2],
\end{align*}
which is implemented using an additional discrete decision variable $v\in\{0,1,2\}$, corresponding to the number of active chillers. 
 The state is subject to time-varying comfort bounds centred around $T=0$:
\begin{align*}
T_{\min}(t) \leq T\leq T_{\max}(t).
\end{align*}
The corresponding discrete-time system is given by
\begin{align*}
x(t+1)=A x(t)+B u(t)+e(t),~(x(t),u(t))\in\mathcal{Z}(t),
\end{align*}
with $x=T$, $u=q$ and periodically time-varying $e(t)$ based on $q_{\text{amb}}$, $T_{\text{amb}}$. 
The economic cost function is to minimize the electricity cost\footnote{%
The proposed framework can also consider peak-demand prices using the formulation in~\cite{risbeck2019economic} or unpredictably changing constraint sets (reflecting comfort levels set by a user) using soft constraints (c.f. Remark~\ref{rk:price}). 
} given by
$\ell(x,u,t,y)=\rho(t,y)\cdot u$, with the price profile $\rho(t,y)$.

In~\cite{risbeck2019economic}, for fixed periodic price signals $\rho(t)$, it was shown that periodic economic formulations~\cite{zanon2017periodic,alessandretti2016convergence,risbeck2019unification} outperform tracking formulations~\cite{limon2014single,JK_periodic_automatica}. 
We consider the more challenging problem, where the price profile $\rho$ changes each day. 
Furthermore, we assume that only the price profile for the current day is available as a forecast, which is modelled using the external variable $y$  that changes every 24 hours. 
The considered  price  profile $\rho$ is taken from the real data considered in~\cite{rawlings2018economic} over the span of one week. 

We implemented the proposed approach using $N=2$, $\beta=10$ with the modified cost from Prop.~\ref{prop:av_perf_2}  and a terminal equality constraint (TEC), which also satisfies the properties in Lemma~\ref{lemma:change_r} for the considered scalar stable system. 
The resulting MPC optimization problem~\eqref{eq:MPC} is a small scale mixed-integer linear program (MILP), which was solved to optimality using \textit{intlinprog} from \textit{MATLAB}. 
The resulting closed loop can be seen in Figure~\ref{fig:HVAC}. 
The closed loop yields a periodic like operation for each day, with small changes between each day based on the different price profile $\rho$. 
The adjustment of the closed-loop response based on the price $\rho$ can be directly seen with the applied input $u$, which is always at a maximum when the electricity price $\rho$ is low. 
We also compared the proposed framework to the optimal\footnote{%
To allow for a consistent comparison, the same initial and final state is considered. 
} operation in Figure~\ref{fig:HVAC}, assuming full knowledge of the future price profile $\rho(t)$ for all coming days.  
The proposed framework results in state and input trajectories very similar to the optimal operation, resulting in a minimal increase of $0.1\%$ in the overall electricity price.

\begin{figure}[hbtp]
\begin{center}
\includegraphics[width=0.41\textwidth]{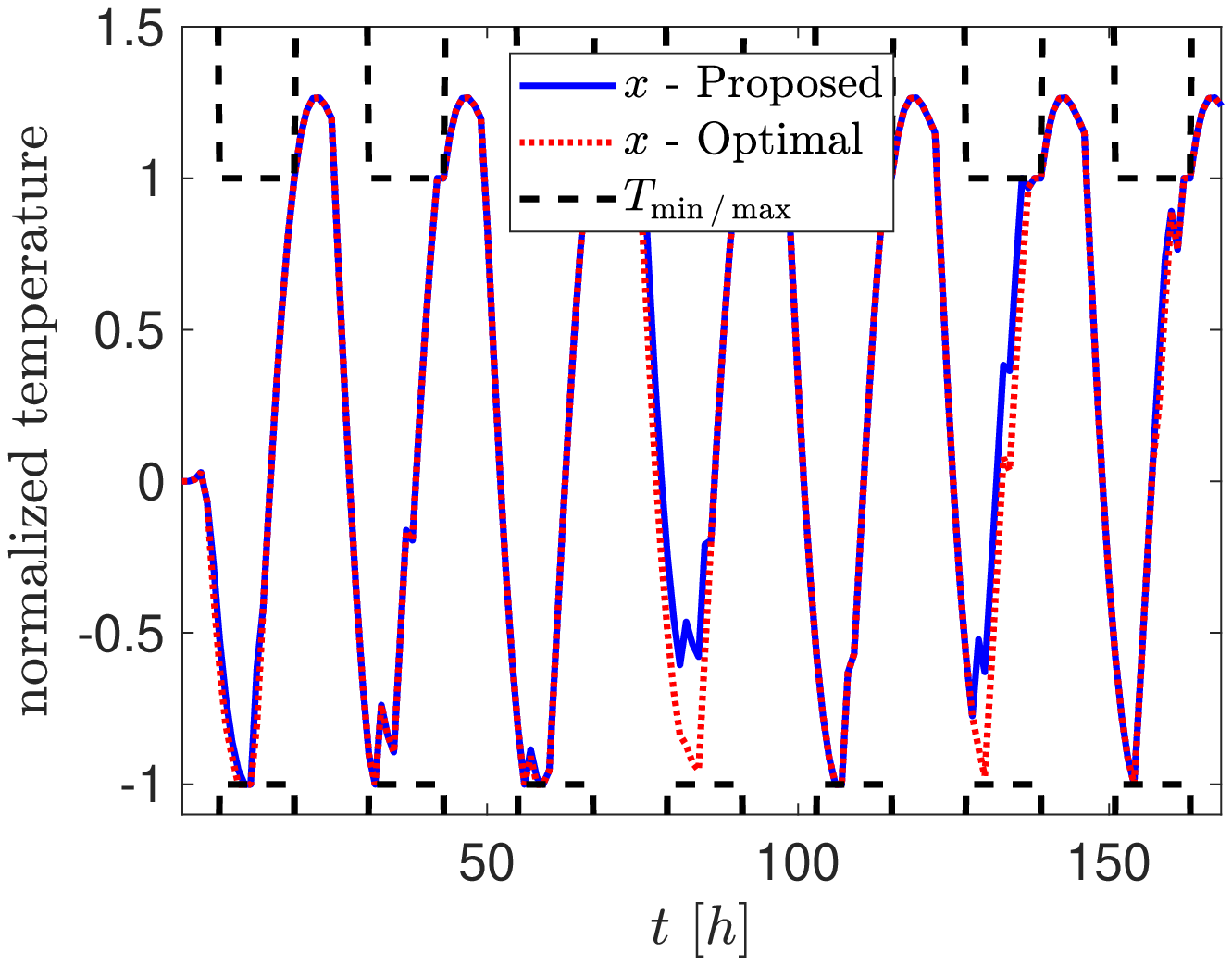}
\includegraphics[width=0.4\textwidth]{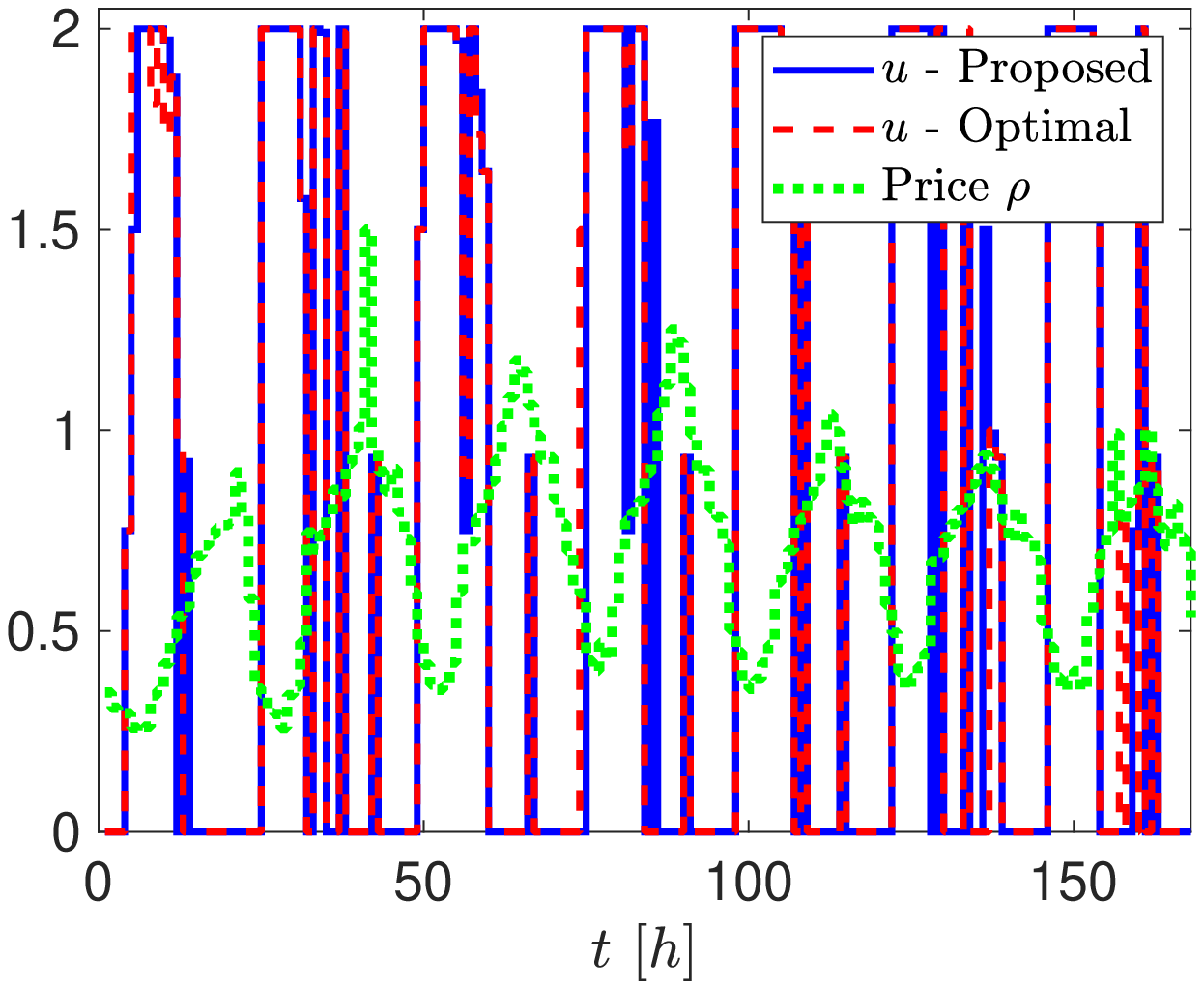}
\end{center}
\caption{%
Transient performance under online changing price signals $\rho$ for HVAC. 
Top: Closed-loop temperature $x=T$ of the proposed approach (blue, solid) and the optimal operation (red, dotted), with time-varying constraints $T_{\min/\max}$ (black, dashed). 
Bottom: Closed-loop applied cooling rate $u=q$ for the proposed approach (blue, solid) and optimal operation (red, dotted); and the price signal $\rho$ (green, dotted).  
}
\label{fig:HVAC}
\end{figure}

\subsection{Continuous stirred-tank reactor}
\label{sec:CSTR}
In the following time-invariant example, we first demonstrate average performance improvement of the proposed framework compared to fixed periodic operation or steady-state operation. 
Then we show reliable economic performance under online changing dynamic operation due to changing cost functions. 
\subsubsection*{System model}
We consider a continuous-time model of a continuous stirred tank reactor (CSTR)
\begin{align*}
\begin{pmatrix}
\dot{x}_1\\\dot{x}_2\\\dot{x}_3
\end{pmatrix}
=&\begin{pmatrix}
1-x_1-10^4x_1^2\exp\left(\frac{-1}{x_3}\right)-400x_1\exp\left(\frac{-0.55}{x_3}\right)\\
10^4 x_1^2\exp\left(\frac{-1}{x_3}\right)-x_2\\
u-x_3
\end{pmatrix},
\end{align*}
where $u\in\mathbb{R}$ is related to the heat flux and $x=(x_1,~x_2,~x_3)^\top\in\mathbb{R}^3$ correspond to the concentration of the reaction, the desired product and the temperature, compare~\cite{bailey1971cyclic,muller2014convergence}, \cite[Sec.~3.4]{faulwasser2018economic}. 
The constraints are 
\begin{align*}
\mathcal{Z}_r=&[0.05,0.4]\times[0.05,0.2]^2\times[0.059,0.439],\\
\mathcal{Z}=&[0.03,1]^3\times[0.049,0.449],
\end{align*}
 and we consider the economic stage cost $\ell(x,u,y)=-x_2+y\cdot(u-u_s)^2$ with $u_s=0.1491$ and $y\in \mathbb{Y}=[0,1]$.
If the external parameter is $y=0$, the stage cost $\ell$ tries to maximize the production of the desired product $x_2$. 
The online tunable part in the cost function is a regularization of the input $u$ relative to the optimal steady-state input $u_s$. 
 For $y=1$ the system is optimally operated at a steady-state $(x_s,u_s)$, while for $y=0$ dynamic operation  significantly outperforms steady-state operation, compare~\cite[Sec.~3.4]{faulwasser2018economic}. 
Hence, treating $y$ as an external variable allows a user to smoothly transition between steady-state and dynamic operation. 
The discrete-time model is defined with a fourth order Runge-Kutta discretization and a sampling\footnote{
In~\cite{muller2014convergence,faulwasser2018economic} a sampling time of $h=0.1$ is used.  
However, with the considered fourth order explicit Runge-Kutta discretization, a sampling time of $h=0.1$ does not preserve stability of the continuous-time system. 
In addition, we consider $x_i\geq 0.03$ instead of $x_i\geq 0$, to avoid discretization errors for $x_i\approx 0$. } time of $h=0.05$. 
For the following simulations, the initial condition is always chosen as the optimal steady-state. 

\subsubsection*{Average performance improvement}
We first consider the problem of maximizing the concentration $x_2$ ($y=0$) to show average performance improvements. 
In the absence of transient changes ($y$ constant), the average performance of periodicity constraint MPC~\cite{houska2017cost,wang2018economic} and tracking formulations~\cite{limon2014single,JK_periodic_automatica} are equivalent (assuming that convergence is achieved).
Similarly, the proposed framework yields the same asymptotic performance as~\cite{zanon2017periodic,alessandretti2016convergence,risbeck2019unification} assuming that $r_T$ converges. 

We implemented the proposed approach with $T\in\{1,10,20\}$ and horizons $N\in[1,50]$ and tested different proposed designs regarding the terminal ingredients ($V_f,\mathcal{X}_f$,Corollary~\ref{corol:term_generic}, Remark~\ref{rk:combine_linear_quadratic_pdf}, Lemma~\ref{lemma:change_r_TEC})  and the cost function ($\tilde{V}_f$, Lemma~\ref{lemma:term_cost_mod}, Prop.~\ref{prop:av_perf_2}). 
The detailed numerical results for all the considered implementations can be found in Appendix~\ref{app:CSTR}. 
In the following, we only consider the approach utilizing the positive definite terminal cost $V_f$ from Remark~\ref{rk:combine_linear_quadratic_pdf} (based on Alg.~\ref{alg:offline}) in combination with the modified cost $\tilde{V}_f$ from Lemma~\ref{lemma:term_cost_mod}, which seems most suitable for practical applications (in terms of computational complexity and performance).
Figure~\ref{fig:EcoPerform} exemplarily shows the performance of this approach with $T\in\{1,10,20\}$ for increasing $N$ in comparison to the average cost at the optimal periodic orbit of length $T=\{10,20\}$ and the optimal steady-state ($T=1$).   
We note that, neglecting small initial deviations\footnote{%
The average performance is computed in the interval $t\in[1000,2000]$ starting with initial condition $x=x_s$. 
Thus, for very short horizons $N$, $\kappa$ has not yet converged.}, 
the proposed EMPC outperforms optimal periodic operation with the same period length $T$, even though a constant value $\beta$ is used (Ass.~\ref{ass:B1_2} does not hold).
We can see in general that the performance increases (for both purely periodic operation and the proposed approach) if we increase $T$ and $N$. 
This implies that the proposed framework utilizing periodic orbits ($T>1$) and additional predictions ($N\geq 1$) with a purely economic formulation can outperform periodicity constrained formulations~\cite{houska2017cost,wang2018economic} ($N=0$), steady-state formulations~\cite{muller2013economic,muller2014performance,fagiano2013generalized,ferramosca2014economic} ($T=1$) and periodic  tracking formulations~\cite{limon2014single,limon2016mpc,limon2018nonlinear,JK_periodic_automatica} ($\ell$ positive definite), even in case of fixed optimal operation ($y$ constant).

\begin{figure}[hbtp]
\begin{center}
\includegraphics[width=0.45\textwidth]{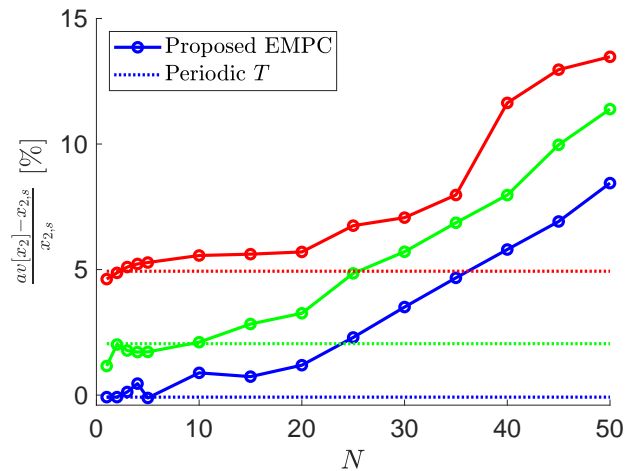}
\end{center}
\caption{
Average performance improvement due to dynamic operation relative to the optimal steady-state $x_s$ - CSTR. 
Periodic operation (dotted) vs. proposed economic MPC scheme (solid with circles) for $T=1$ (blue), $T=10$ (green) and $T=20$ (red).}
\label{fig:EcoPerform}
\end{figure}

\subsubsection*{Transient performance under online changing conditions}
  \begin{figure}[hbtp]
\begin{center}
\includegraphics[width=0.45\textwidth]{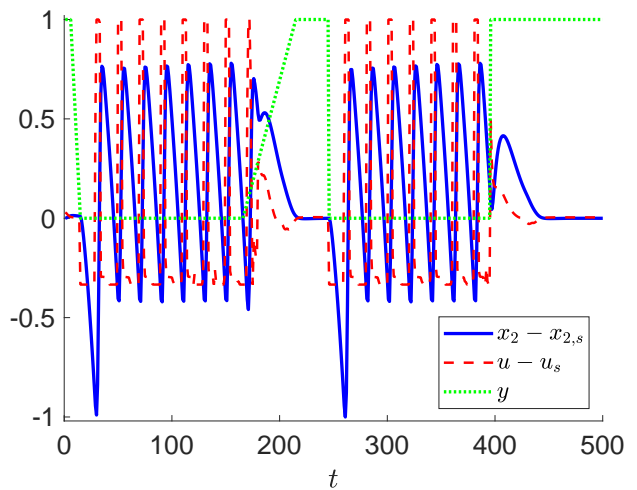}
\end{center}
\caption{Dynamic operation under online changing conditions: deviation in production $x_2-x_{2,s}$ (blue, solid), deviation in heat flux $u-u_s$ (red, dashed), price signal $y$ (green, dotted). All signals are normalized to $|x|\leq 1$. }
\label{fig:dynEcoPerform}
\end{figure}
In the following, we study the performance of the proposed scheme under online changing conditions, i.e., $y_t$ unpredictably time-varying. 
The resulting closed loop for the proposed MPC scheme with $N=10$, $T=20$ can be seen in Figure~\ref{fig:dynEcoPerform}.
As $y_t\rightarrow 0$ (e.g. $t=15$  or $t=246$), the system operates dynamically to increase production $x_2$ and once  the weight $y$ on input deviations increases the system quickly minimizes the control effort and smoothly converges (e.g. $t\in[185,246]$, $t\in[400,470]$) to the new optimal mode of operation, i.e.,  the steady-state $x_s$. 
In this scenario, the MPC scheme on average still increases production by $2.8\%$ compared to steady-state operation, while a $5\%$ increase was achieved for $y\equiv 0$ (c.f. Fig.~\ref{fig:EcoPerform}). 
For comparison, the performance of the proposed approach, the tracking MPC~\cite{limon2014single,JK_periodic_automatica} with $T=20$, $N\in[0,50]$  and the periodicity constraint MPC~\cite{houska2017cost,wang2018economic} with $T\in[0,90]$ can be seen in Figure~\ref{fig:dynEcoPerform2}. 
First, note that the number of decision variables in a condensed formulation are $n+m\cdot (T+N)$ for the proposed approach and the tracking MPC~\cite{limon2014single,JK_periodic_automatica},  and $m\cdot T$ for the periodicity constraint MPC~\cite{houska2017cost,wang2018economic} ($N=0$). 
Thus, the x-axis ($N+T$) in Figure~\ref{fig:dynEcoPerform2} is a measure for the computational complexity. 
First, note that we can further improve the performance of the proposed approach by increasing $N$. 
Similarly, the performance of the periodicity constraint MPC~\cite{houska2017cost,wang2018economic} improves for a larger period length $T$, but at a smaller pace.
Thus, given the same number of decision variables, the proposed approach can achieve a better performance. 
If we consider the tracking MPC, for small values of $N$, the performance is similar to the proposed economic formulation. 
However, in contrast to the proposed formulation, the economic performance of the tracking MPC does not improve significantly with a large horizon $N$. 
Additional numerical results can be found in Appendix~\ref{app:CSTR}. 

To summarize, in the considered example we have shown the applicability of the proposed approach to nonlinear economic optimal control problems. 
In particular, the proposed approach:
(i) improves performance compared to (fixed) steady-state or periodic operation,
(ii) reliably operates under online changing conditions,
(iii) in general achieves better performance than periodicity constraint formulations~\cite{houska2017cost,wang2018economic} or tracking formulations~\cite{limon2014single,JK_periodic_automatica}.

\begin{figure}[hbtp]
\begin{center}
\includegraphics[width=0.45\textwidth]{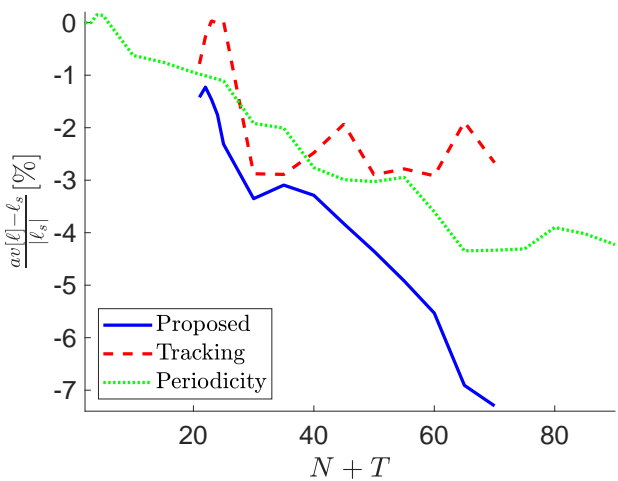}
\end{center}
\caption{%
Transient performance $\ell(x,u,y)=-x_2+y\cdot (u-u_s)^2$ in dynamic operation under online changing conditions for CSTR relative to steady-state operation $\ell_s=\ell(x_s,u_s)$: Proposed approach (blue, solid), tracking MPC~\cite{limon2014single,JK_periodic_automatica} (red, dashed) with $T=20$ and $N\in[0,50]$ and periodicity constraint MPC~\cite{houska2017cost,wang2018economic} (green, dotted) with $T\in[0,90]$.  }
\label{fig:dynEcoPerform2}
\end{figure}

\section{Conclusion}
\label{sec:sum}
We have presented an economic MPC framework that is applicable to nonlinear (periodically) time-varying problems and online changing operation conditions. 
We have shown recursive feasibility, constraint satisfaction and derived performance guarantees relative to periodic (locally) optimal operation. 
Interestingly, the problem of economic periodic operation requires additional techniques compared to optimal steady-state operation, which was also the case in~\cite{muller2016economic}.    
In particular, we used a novel continuity condition on the economic cost of periodic orbits, to reformulate the update scheme and constraints on the artificial periodic reference trajectory. 
In addition, we proposed novel offline computations to obtain suitable terminal ingredients. 
We have demonstrated the applicability and practicality of the proposed framework for economic periodic control with an HVAC and a CSTR example. 
The strength of the framework comes from the fact that many existing schemes~\cite{muller2013economic,muller2014performance,zanon2017periodic,amrit2011economic,alessandretti2016design,angeli2012average} or modified versions thereof~\cite{fagiano2013generalized,houska2017cost,wang2018economic,alessandretti2016convergence,angeli2016theoretical} are contained as special cases.

The practical application of the proposed framework to large scale HVAC systems is part of future work. 
A theoretical analysis of (bounded) online changing constraint sets/dynamics with robust performance guarantees is an open problem.

\bibliographystyle{IEEEtran}  
\bibliography{Literature}  
\clearpage

\appendix
\subsection{Continuous stirred-tank reactor - details}
\label{app:CSTR}
In the following, we provide additional details regarding the implementation of the CSTR example in Section~\ref{sec:CSTR} and investigate the economic performance of the different considered implementations in more detail.
\subsubsection*{Implementation details} 
In the implementation, we consider $c_{\kappa}=100$ $\beta=10$. 
The offline computation is done with an Intel Core i7 using the semidefinite programming (SDP) solver SeDuMi-1.3 \cite{sturm1999using} and the online optimization is done with CasADi~\cite{andersson2012casadi}. 
 The iterations in CasADi were stopped after $1000$ iterations, although typically the standard tolerance was satisfied.
Since the optimization problem~\eqref{eq:MPC} is not solved to optimality but a fixed number of iterations are done online, the resulting reference $r$ is not necessarily a feasible periodic orbit. 
Hence, the primal infeasibility of the reference $r$ should be taken into account when defining $\kappa_j$ in~\eqref{eq:close_3} (e.g. in terms of a penalty factor), to avoid persistent feasibility issues. 
Thus, we replaced \eqref{eq:close_3} by
\begin{align*}
\kappa_j(t+1)=& \ell(r^*_T(j+1|t),t+N+1+j,y(t+1))+10^3 \epsilon,
\end{align*}
 where $\epsilon$ is the largest constraint violation (measured in terms of the infinity norm)  in the periodicity constraint~\eqref{eq:MPC_periodic}. 
In addition,  we replaced the resulting optimized trajectory by the feasible candidate solution in case it has a worse cost (up to $10^{-4}$) or does not satisfy the constraints (up to $10^{-3}$), which ensures that the performance guarantees in Proposition~\ref{prop:av_perf} remain valid independent of numerical issues\footnote{%
While ideally any solver should guarantee that the resulting solution is no worse than any provided feasible initial solution~\cite{scokaert1999suboptimal}, this was not always the case in the considered implementation with IPOPT, probably due to the difficulty of initializing the dual variables.}.

\subsubsection*{Offline computations}
For the design of the terminal cost $V_f$  and the terminal set $\mathcal{X}_f$, we proceed along the lines of Algorithm~\ref{alg:offline}. 
We consider $S=\text{diag}(0,I_m)\succeq\ell_{rr}$ and $2\epsilon+\tilde{\epsilon}=0.1$ in~\eqref{eq:stage_output}. 
The matrices $P(r)$, $K(r)$ are computed using the LMIs from \cite[Lemma~2]{JK_QINF} and gridding $(x_1,x_2,x_3,u,u^+)$ using $20^3\cdot 2=16.000$ points, which takes 37~minutes.
The terminal set is $\mathcal{X}_f=\{x|~\|x-x_r\|_{P(r)}^2\leq \alpha\}$, with $\alpha=1.2\cdot 10^{-8}$.
The verification of $\alpha_1=\alpha_2=1.2\cdot 10^{-8}$ along the lines of \cite[Alg.~1]{JK_QINF} took approximately $45$ minutes.
For this example,  the convex formulation~\cite[Prop.~1]{JK_QINF} can only be used if the sampling time $h$ is reduced and the resulting terminal cost tends to be very conservative. 
For details on the different formulations, compare also~\cite[Example~1]{JK_QINF}. 
In addition, we computed a constant $c\approx 21$ offline, such that the simpler terminal cost proposed in Remark~\ref{rk:combine_linear_quadratic_pdf} also satisfies Assumptions~\ref{ass:term_simple} and \ref{ass:term}.  

 \subsection*{Average performance improvement}
We first consider the general economic performance with the fixed economic stage cost $\ell(x,u)=-x_2$, using $y=0$ and compare the performance relative to the performance at the optimal steady-state $x_s=(x_{1,s},x_{2,s},x_{3,s})=(0.0832;0.0846;0.1491)$. 
 The performance is specified as improvement (in $\%$) of the average concentration $x_2$ and is computed over the time-interval $t\in[1000,2000]$, where the first $1000$ steps are discarded to better reflect the asymptotic average performance.

 First, Figure~\ref{fig:CSTR_App1} shows that we can  in general achieve better performance by using a dynamic operation. 
 We can see that increasing the length $T$ of the optimal $T$-periodic orbit \eqref{eq:opt_periodic} can increase the average production by approximately $8\%$. 
 We would like to point out, that the ``optimal'' T-periodic orbits where computed using CasADi and may thus also represent local minima. 
In addition to the periodic orbit, we also considered and implemented an economic MPC scheme without any terminal ingredients or artificial periodic trajectories~\cite{grune2013economic,muller2016economic,grune2018economic}. 
We would first like to point out again, that in general for this setup (nonlinear system, unknown mode of optimal dynamic operation), guaranteeing a priori desired performance bounds with such an implementation  is difficult. 
Nevertheless, this implementation supports two important arguments of the proposed approach. 
First, we can see that for short horizons $N$ (e.g. $N\leq 20$ here) the performance of such a simple implementation without suitable modifications can be worse than steady-state performance and thus worse than implementing a simple tracking/stabilizing MPC.
Second, we can see for large horizons $N$, e.g. $N\geq 50$, that it is possible to obtain significantly better performance than periodic operation with a fixed period length $T$, which can be obtained with periodic tracking MPC schemes like~\cite{limon2014single,limon2016mpc,JK_periodic_automatica}. 
We also included a scheme with multi-step implementation, since there exist theoretical reasons to expect better performance if the \textit{optimal period length} is used~\cite{muller2016economic}.
However, as it is not clear if the system is optimally periodically operated and/or with which period length, we simply implemented a small value $\nu=3$, which, however, did not show any performance benefits (except for the reduced computational demand).

\begin{figure}[hbtp]
\begin{center}
\includegraphics[scale=0.45]{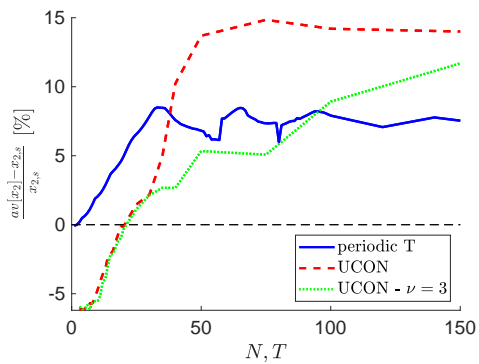}
\end{center}
\caption{Average performance improvement due to dynamic operation relative to optimal steady-state $x_s$ - CSTR. 
Periodic orbit (blue, solid, periodic) with period length $T$, unconstrained economic MPC (red, dashed, UCON) with horizon $N$ and unconstrained economic MPC with multi-step implementation using $\nu=3$ (green, dotted, UCON - $\nu=3$).}
\label{fig:CSTR_App1}
 \end{figure}

We implemented the proposed approach (Problem~\eqref{eq:MPC}) using $N\in[0,50]$, $T\in\{1,10,20\}$, $\beta=10$ and\footnote{%
We did not formally verify satisfaction of Assumption~\ref{ass:cont_orbit} with $c_{\kappa}=100$. Instead ,we simply checked numerically that the proposed approach was able to converge to (local) optimal periodic orbits $r_T$ with the chosen value.
} $c_{\kappa}=100$. 

For the terminal ingredients; we implemented the design in  Corollary~\ref{corol:term_generic} (QINF), the positive definite terminal cost $V_f$ from Remark~\ref{rk:combine_linear_quadratic_pdf} (QINF-pdf), and a terminal equality constraint (TEC) with $\nu=3$ (Lemma~\ref{lemma:change_r_TEC}).
The corresponding results can be seen in Figure~\ref{fig:CSTR_App2}.
 For the economic terminal cost $V_f$ from Corollary~\ref{corol:term_generic}, the constraints~\eqref{eq:compute_p} to compute $p(\cdot|t)$ are implemented using~\eqref{eq:con_p_XY} from Remark~\ref{rk:compute_p}.
 We also tested the alternative explicit formulation~\eqref{eq:direct_p}, however, except for the trivial case $T=1$, this lead to numerical difficulties and a significant increase in online iterations without any major benefit. 
 We also implemented the terminal equality constraint (TEC) MPC using a one-step implementation ($\nu=1$), which resulted in an improved performance for $T,N$ large.

If we use an artificial setpoint $(T=1$, top figure), we can a) clearly see that the performance further improves if we increase the prediction horizon $N$; b) we can see that a terminal cost/set (QINF) improves the performance relative to a terminal equality constraint (TEC).
For example, the performance with $N=50,T=1$ with TEC is worse than the performance with a terminal cost using $N=40$, $T=1$, thus showing the potential for performance improvement/computational saving using suitable terminal ingredients. 
On the other hand, the difference between the positive definite  terminal cost (QINF-pdf) and the economic terminal cost (QINF) is often small and thus the simple design in Remark~\ref{rk:combine_linear_quadratic_pdf} may be most favourable for practical implementations. 
For $T=10$ (middle figure) and $T=20$ (bottom figure) we can see similar performance differences for larger horizons $N$: We can achieve better performance with a large horizon $N$ or larger period length $T$, and we can achieve better performance if we use a terminal cost (QINF, QINF-pdf) instead of a terminal equality constraint (TEC).

\begin{figure}[hbtp]
\includegraphics[width=0.4\textwidth]{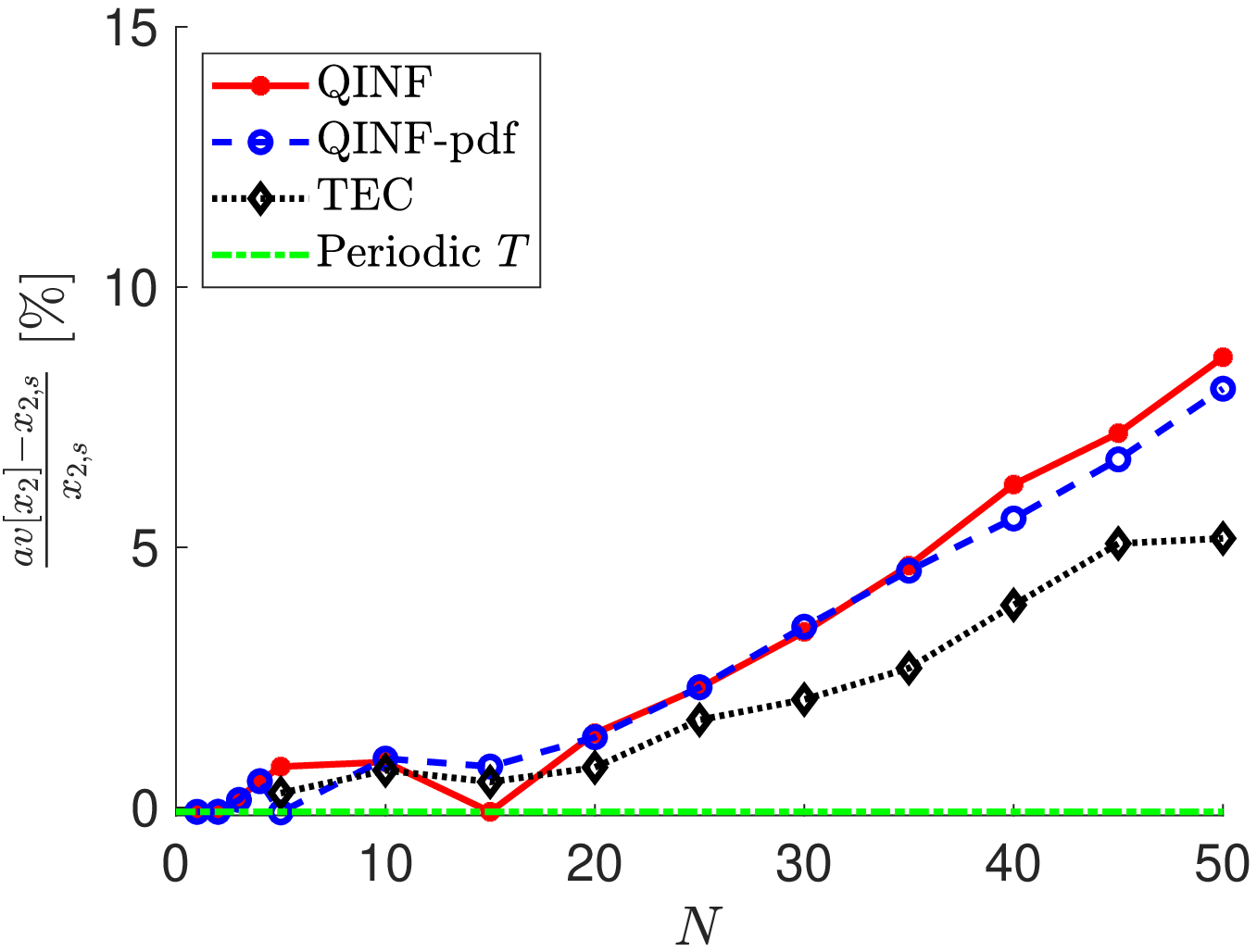}\\
\includegraphics[width=0.4\textwidth]{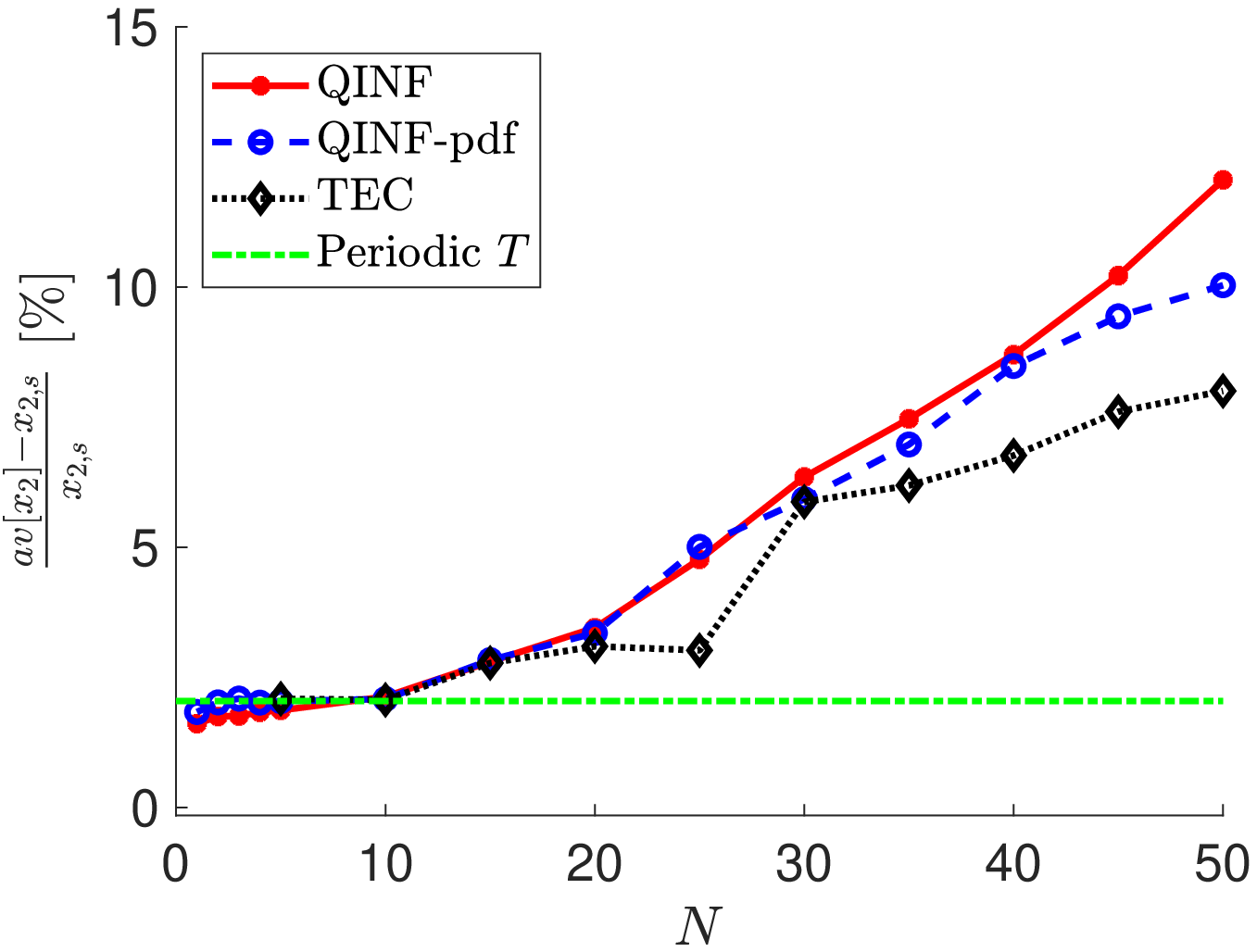}\\
\includegraphics[width=0.4\textwidth]{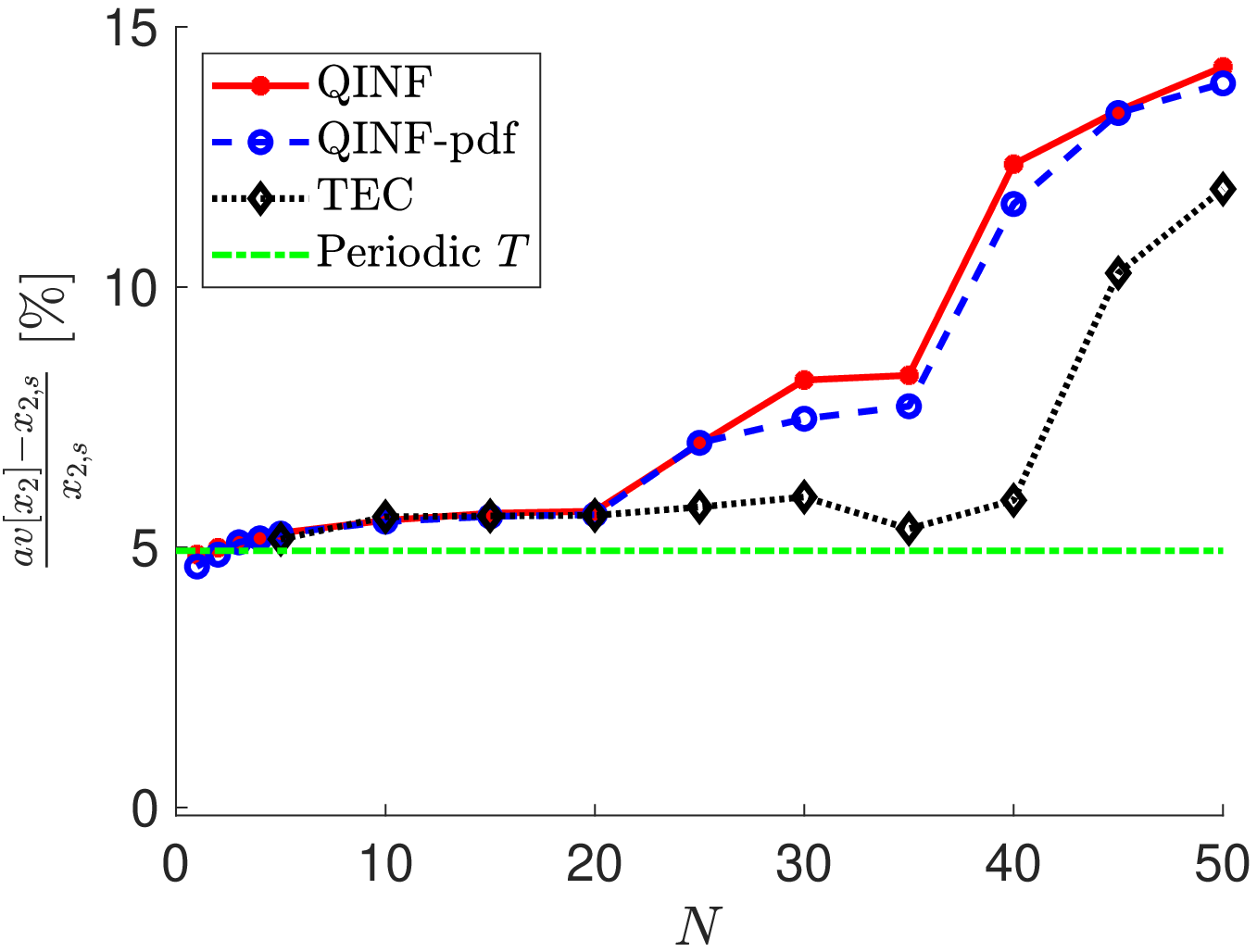}
\caption{Average performance improvement relative to optimal steady-state $x_s$ for $T=1$ (top), $T=10$ (middle) and $T=20$ (bottom). Approach with economic terminal cost (red, solid, QINF); positive definite terminal cost (blue, dashed, QINF-pdf);  terminal equality constraint (black, dotted, TEC) with prediction horizon $N\in[1,50]$.
Performance at the optimal $T$-periodic orbit is show in green, dashed for comparison.}
\label{fig:CSTR_App2}
 \end{figure}

We also implemented the different terminal ingredients with the modified cost $\tilde{V}_f$ from Lemma~\ref{lemma:term_cost_mod}, compare Figure~\ref{fig:CSTR_App3}. 
If we compare the results in Figure~\ref{fig:CSTR_App2} and Figure~\ref{fig:CSTR_App3}, the performance is essentially similar. 
Thus, the modified formula $\tilde{V}_f$ (Lemma~\ref{lemma:term_cost_mod}) does not seem to have any particularly negative impact on the performance, while reducing the computational demand by  replacing the constraints~\eqref{eq:kappa_con_1}--\eqref{eq:kappa_con_2} with \eqref{eq:kappa_con_new}.

\begin{figure}[hbtp]
\includegraphics[width=0.4\textwidth]{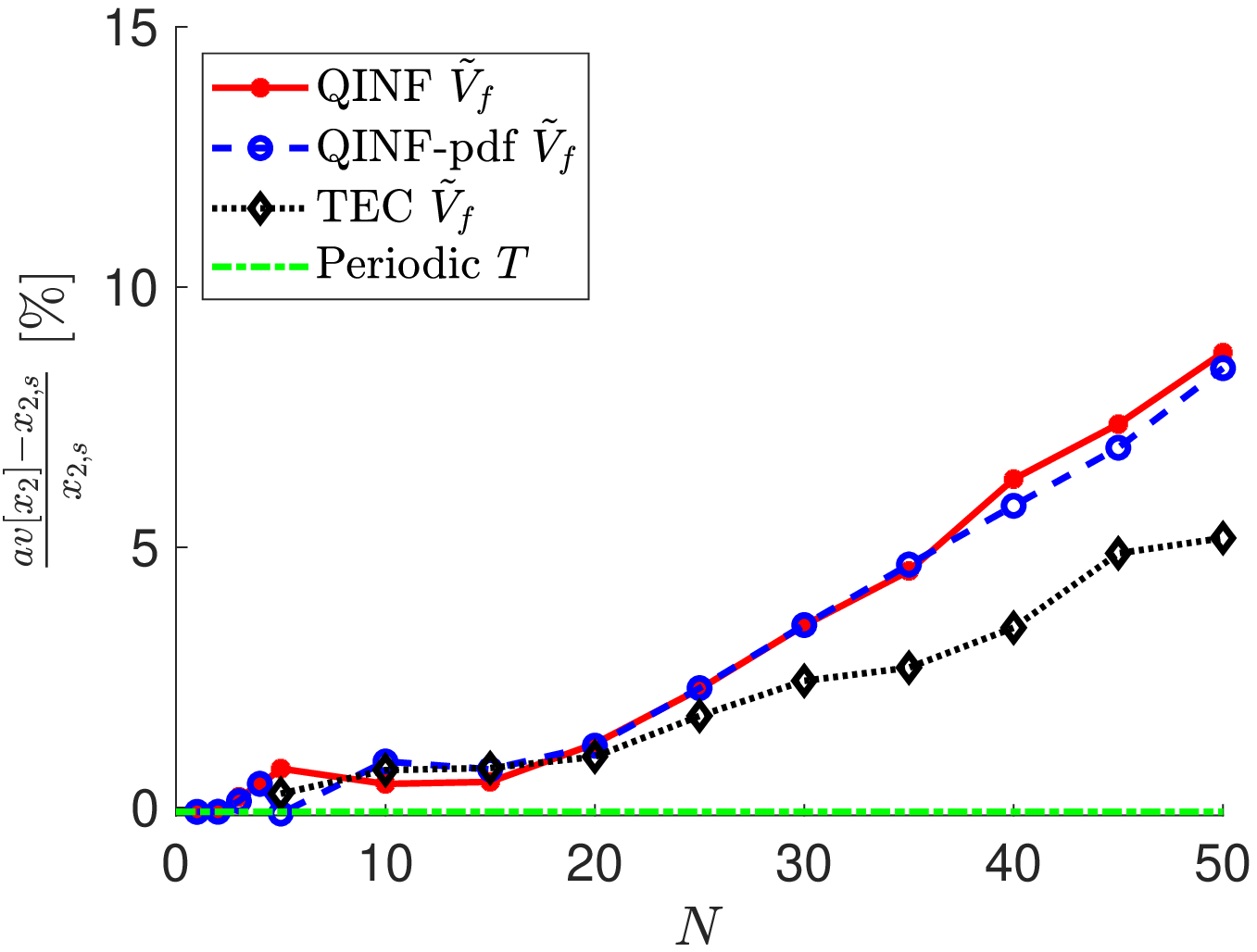}\\
\includegraphics[width=0.4\textwidth]{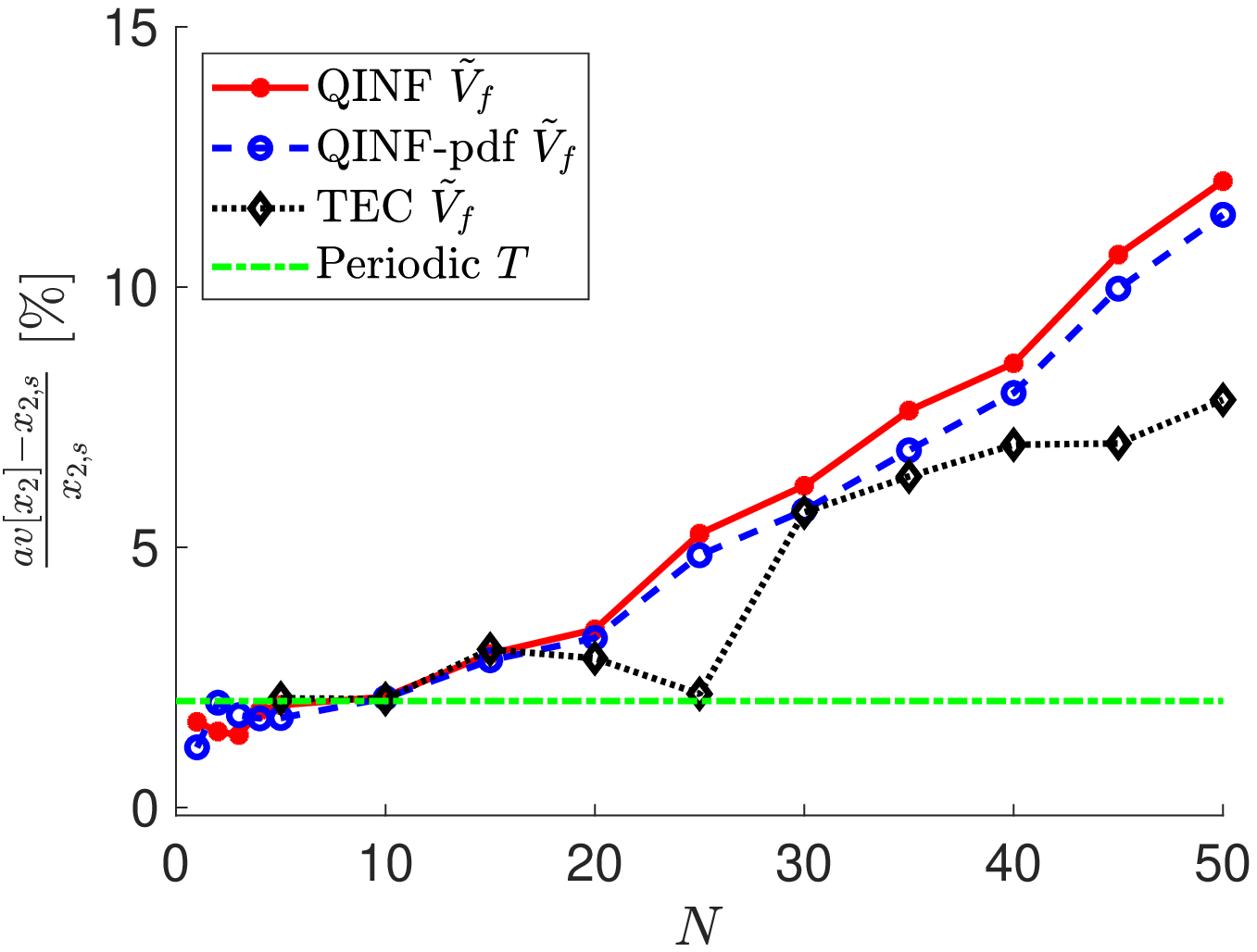}\\
\includegraphics[width=0.4\textwidth]{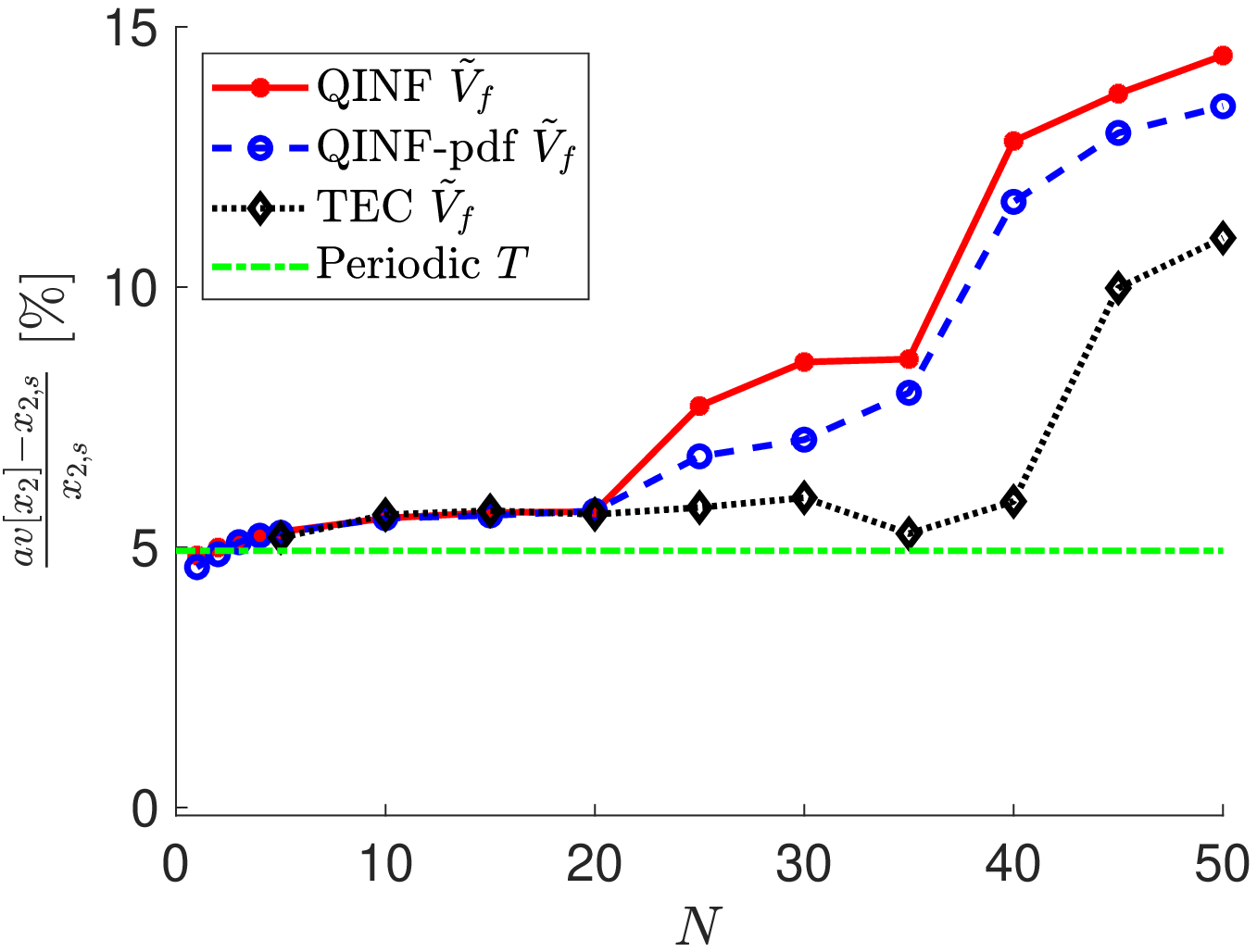}
\caption{%
 Average performance improvement relative to optimal steady-state $x_s$ for $T=1$ (top), $T=10$ (middle) and $T=20$ (bottom) using the modified cost $\tilde{V}_f$. Approach with economic terminal cost (red, solid, QINF); positive definite terminal cost (blue, dashed, QINF-pdf);  terminal equality constraint (black, dotted, TEC) with prediction horizon $N\in[1,50]$.
 Performance at the optimal $T$-periodic orbit is show in green, dashed for comparison. }
\label{fig:CSTR_App3}
 \end{figure}

We also implemented the continuous-time formulation discussed in Remark~\ref{rk:cont_time} with a variable sampling time $h\in[h_{\min},h_{\max}]=[10^{-2},5\cdot 10^{-2}]$. 
 The resulting performance for the  optimal $T$-periodic orbit \eqref{eq:opt_periodic}, and an economic MPC scheme without any terminal ingredients (UNCON) can be seen in Figure~\ref{fig:CSTR_App4}. 
 For comparison, we also implemented the proposed approach with a terminal equality constraint (TEC, Lemma~\ref{lemma:change_r_TEC}), $\nu=3$ and $T=20$.
 We can see that the variable sampling time $h$ significantly improves the overall performance compared to the fixed sampling time implementation shown in Figure~\ref{fig:CSTR_App1}. 
Obtaining improved terminal ingredients similar to Corollary~\ref{corol:term_generic} for this scenario, however, requires further research\footnote{%
As the continuous-time analogue to \eqref{eq:P_ineq} in Prop.~\ref{prop:term_generic}, we can design a quadratic term satisfying $\frac{d}{dt} \|x-x_r\|_{P(r_T,t)}^2\leq -\|x-x_r\|_{Q^*(r_T,t)}^2$ using   \cite[Lemma~4/Prop.~5]{JK_QINF}, which only requires $10$s and $53$s, respectively (compared to $37$ min for the discrete-time formulation).
However, a suitable terminal cost $V_f$ also requires a simple means to compute a corresponding gradient correction term $p$ in continuous-time. 
In addition, due to the average cost function proposed in Remark~\ref{rk:cont_time}, we need $\frac{dV_f}{dt}\leq -\ell_c/h$, which might only hold by scaling the terminal cost conservatively using $\frac{1}{h_{\min}}$. 
}.

We also tested the naive implementation discussed in Section~\ref{sec:pitfal}.
Although this approach is in general not satisfactory (as shown in Sec.~\ref{sec:pitfal}), in the considered example this approach often performs similar to the proposed formulation.
However, for small horizons $N$ and a small weighting $\beta$, neglecting the constraints~\eqref{eq:kappa_con_1}--\eqref{eq:kappa_con_2} can yield worse performance than steady-state operation, similar to the implementation without terminal ingredients (UCON) in Figure~\ref{fig:CSTR_App1}.

\begin{figure}[hbtp]
\includegraphics[width=0.4\textwidth]{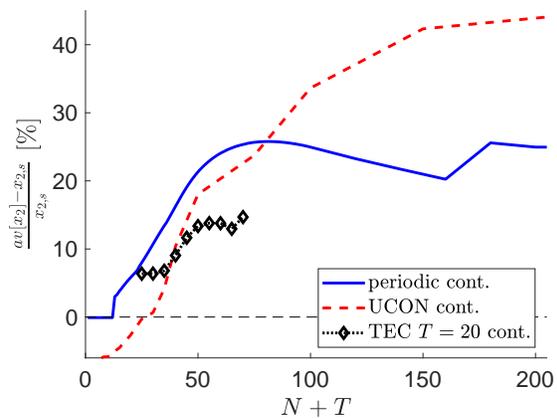}
\caption{Average performance improvement relative to optimal steady-state $x_s$ with flexible sampling time $h$. 
Periodic orbit (blue, solid, periodic cont.) with period length $T$, unconstrained economic MPC (red, dashed, UCON cont.) with horizon $N$ and terminal equality constraint using $\nu=3$ and $T=20$ (green, dotted, TEC $T=20$ cont.)}
\label{fig:CSTR_App4}
 \end{figure} 

\subsubsection*{Summary}
\textbf{Improved performance: }
In this scenario, we have only considered the case of fixed known stage cost $\ell$ ($y$ constant). 
In this case, the proposed framework reduces to the periodic economic MPC approaches in~\cite{zanon2017periodic,alessandretti2016convergence,risbeck2019unification}, with a fixed periodic trajectory $r_T$ (assuming convergence of $r_T$ and considering the special case of periodic operation in~\cite{alessandretti2016convergence}). 
We have seen that the proposed fully economic formulation yields a better performance than operating the system at the optimal steady-state ($T=1$) or optimal periodic orbit $(T>1$). 
 In particular, this implies that the proposed framework utilizing periodic orbits ($T>1$) and additional predictions ($N\geq 1$), both with the economic stage cost $\ell$, outperforms periodicity constrained formulations~\cite{houska2017cost,wang2018economic} ($N=0$), steady-state formulations~\cite{muller2013economic,muller2014performance,fagiano2013generalized,ferramosca2014economic} ($T=1$) and periodic  tracking formulations~\cite{limon2014single,limon2016mpc,limon2018nonlinear,JK_periodic_automatica} ($\ell$ positive definite), even in case of fixed optimal operation ($y$ constant).
 The second scenario investigates this difference in performance for the transient case of changing operation conditions ($y$ not constant). 

\textbf{Different formulations: } 
A second key finding in this scenario is the comparison of the different design choices possible in the proposed framework, in particular the terminal ingredients ($V_f,\mathcal{X}_f$,  Ass.~\ref{ass:term_simple}, \ref{ass:term} Corollary~\ref{corol:term_generic}, Remark~\ref{rk:combine_linear_quadratic_pdf}, Lemma~\ref{lemma:change_r_TEC}) and the modified cost function ($\tilde{V}_f$, Lemma~\ref{lemma:term_cost_mod}, Prop.~\ref{prop:av_perf_2}). 

First, even though the cost formulation $\tilde{V}_f$ proposed in Lemma~\ref{lemma:term_cost_mod} may be unconventional, we have not found any disadvantages in the performance comparison. At the same time, this formulation enjoys stronger theoretical properties (does not requires Ass.~\ref{ass:cont_orbit}) and is easier to implement ($T$ nonlinear constraints~\eqref{eq:kappa_con_1}--\eqref{eq:kappa_con_2} are replaced by one constraint~\eqref{eq:kappa_con_new}). 

Regarding the terminal cost $V_f$, we can see that a properly designed terminal cost $V_f$ (Corollary~\ref{corol:term_generic}, Remark~\ref{rk:combine_linear_quadratic_pdf}) allows us to achieve a better performance with a shorter horizon $N$ compared to a simple terminal equality constraint (TEC, Lemma~\ref{lemma:change_r_TEC}). 
However, the difference in performance between the terminal cost $V_f$ in   Corollary~\ref{corol:term_generic} and Remark~\ref{rk:combine_linear_quadratic_pdf} seems rather small, while the design in Corollary~\ref{corol:term_generic} requires $T\cdot n$ additional nonlinear constraints~\eqref{eq:con_p_XY}, which does not seem appropriate for $T>>1$ given the small performance improvement. 

Given this comparison, the most suitable formulation seems to be the positive definite terminal cost $V_f$ from Remark~\ref{rk:combine_linear_quadratic_pdf} combined with the modified cost $\tilde{V}_f$ in Lemma~\ref{lemma:term_cost_mod}, although the simple terminal equality constraint may also be appropriate if the design in Alg.~\ref{alg:offline} is not applicable. 
In addition, further investigations into continuous-time formulations for dynamic operations along the lines of Remark~\ref{rk:cont_time} may yield significant performance improvements.

\subsection*{Transient performance under online changing conditions}
In the following, we provide additional details regarding the transient performance comparison.  
The tracking MPC uses the quadratic part of the economic terminal cost from Prop.~\ref{prop:term_generic} $V_{f,tr}(x,r)=\|x-x_r\|^2_{P(r)}$ (similar to~\cite{JK_QINF}) and the quadratic tracking cost $\ell_{tr}(x,u,r)=\|x-x_r\|_Q^2+\|u-u_r\|_R^2$ with $Q=0.05$, $R=1$. 

First, we consider $T=20$ and $N=10$ as in Figure~\ref{fig:dynEcoPerform}. 
In Figure~\ref{fig:dynEcoPerform3}, we see the closed-loop cost $\ell$ in the transition periods $t\in[182,220]$ and $t\in[398,440]$ ($y_t\rightarrow 1$), for the proposed approach, the periodicity-constraint MPC~\cite{houska2017cost,wang2018economic} and the tracking MPC~\cite{limon2014single,JK_periodic_automatica}. 
We can see that all approaches quickly converge to the optimal steady-state.
Crucially, in the transient period, while the system converges to the optimal steady-state, the proposed approach and the tracking MPC~\cite{limon2014single,JK_periodic_automatica} seem to be able to utilize the transition period to further minimize the stage cost $\ell$, while the periodicity constraint MPC~\cite{houska2017cost,wang2018economic} does not seem to have the necessary degrees of freedom. 
The periodicity constraint MPC cannot really take advantage of the transient changes yielding faster convergence but overall small improvement  (only $0.9\%$) compared to steady-state.
For comparison, the proposed approach and the tracking MPC~\cite{limon2014single,JK_periodic_automatica} can improve overall performance by $3.4\%$ and $2.9\%$, respectively. 

Thus, the proposed approach achieves better performance than the periodicity constraint MPC~\cite{houska2017cost,wang2018economic} and the tracking MPC~\cite{limon2014single,JK_periodic_automatica}, even in case of online changing operation conditions ($y_t$)  by utilizing a purely economic formulation with additional degrees of freedom.  

\begin{figure}[hbtp]
\begin{center}
\includegraphics[width=0.4\textwidth]{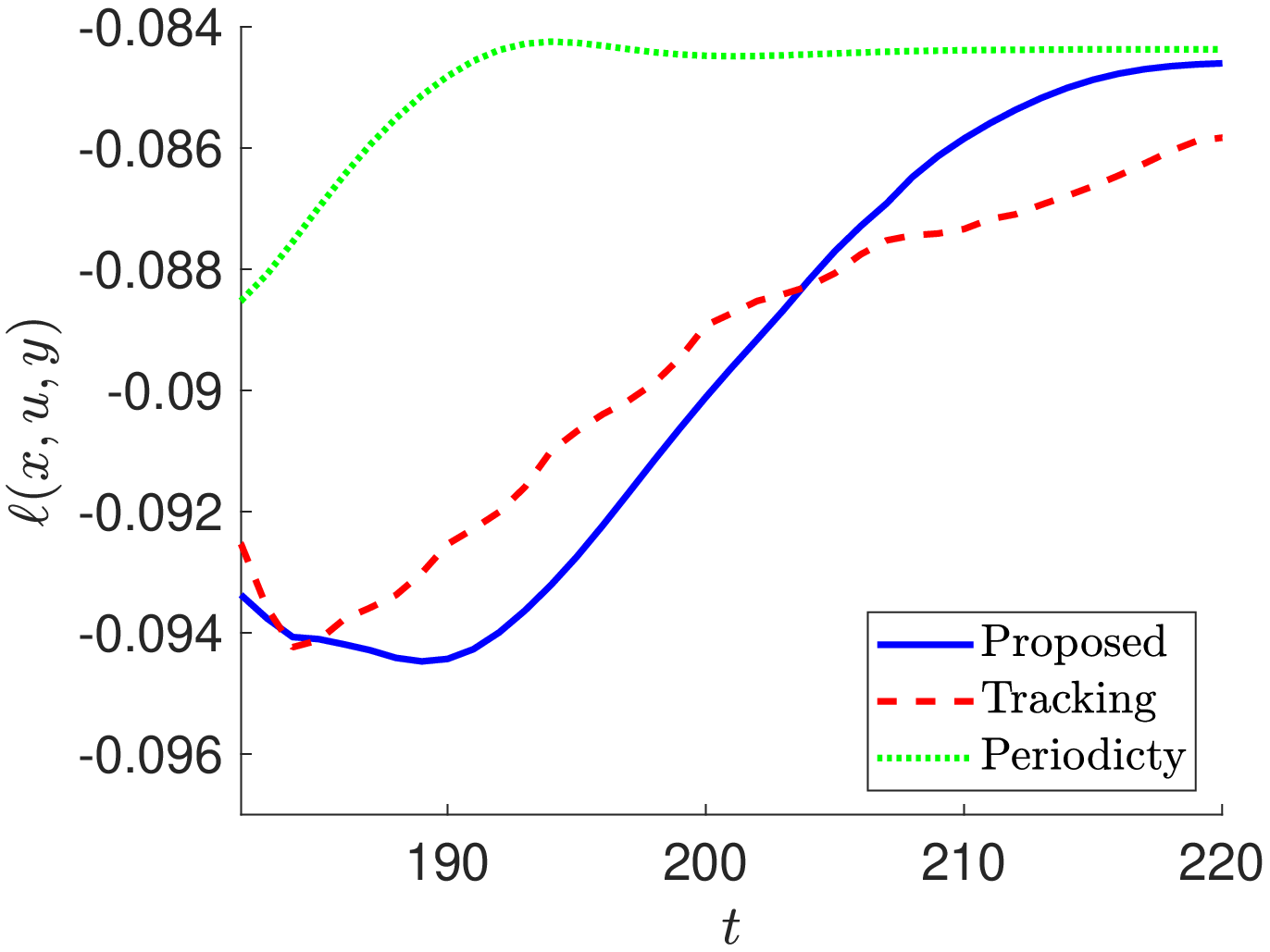}
\includegraphics[width=0.4\textwidth]{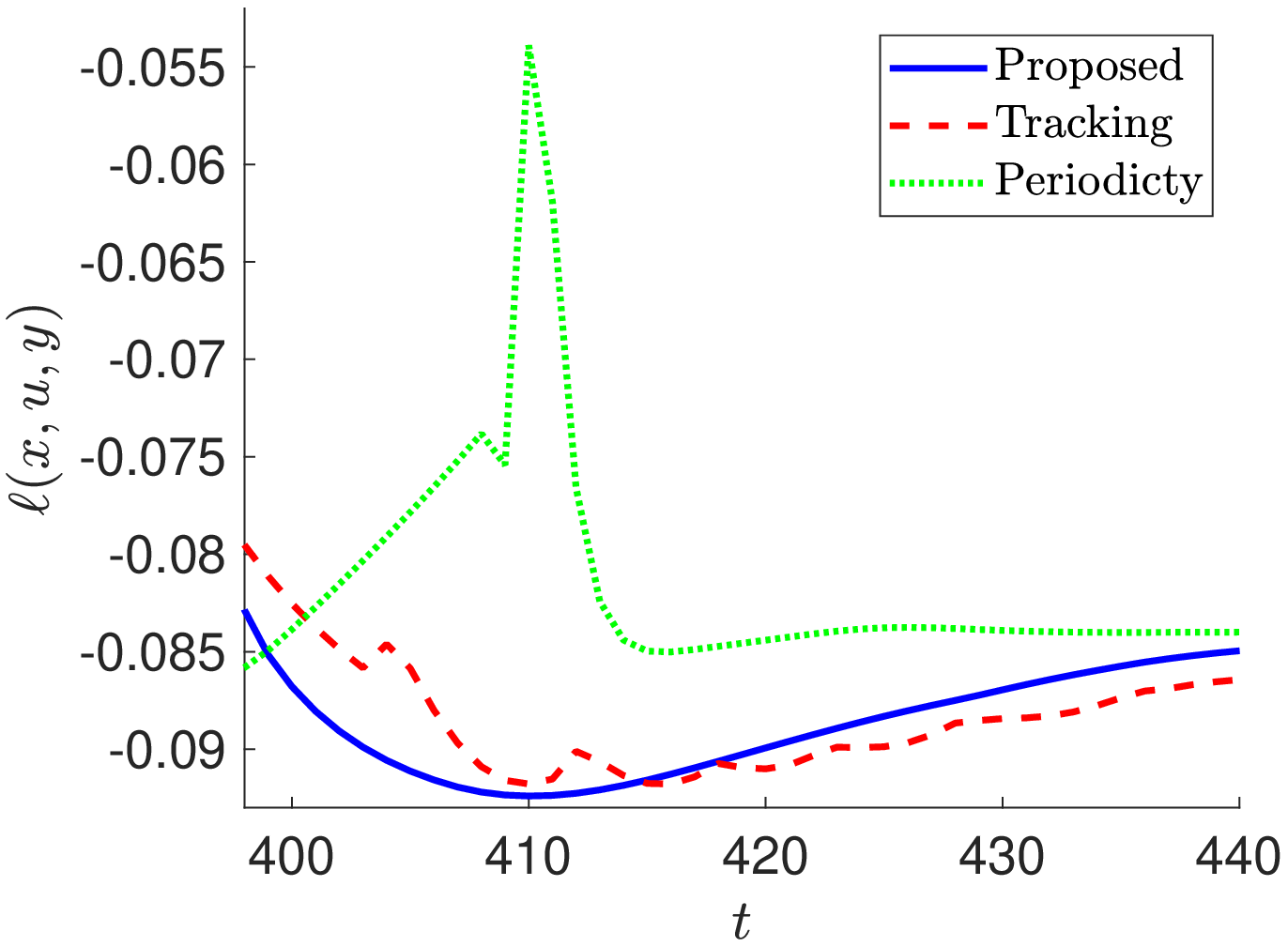}
\end{center}
\caption{%
Transient performance in dynamic operation under online changing conditions for CSTR: Proposed approach (blue, solid), tracking MPC~\cite{limon2014single,JK_periodic_automatica} (red, dashed) with $T=20$ and $N=10$ and periodicity constraint MPC~\cite{houska2017cost,wang2018economic} (green, dotted) with $T=20$ at $t\in[182,220]$ (top) and $t\in[398,440]$ (bottom).  }
\label{fig:dynEcoPerform3}
\end{figure}

\end{document}